\documentclass[a4paper,12pt]{article}
\usepackage[english]{babel}
\usepackage{amsmath,amsthm}%
\usepackage{cite}%
\usepackage{physics}
\usepackage{booktabs}
\usepackage{subfig}
\usepackage{hyperref}
\usepackage{xcolor}
\usepackage{tikz}

\def\Z {\hbox{\Sets Z}}
\def\N {\hbox{\Sets N}}
\def\dfrac#1#2{\frac{\partial #1}{\partial #2}}
\newcounter{transf}

\renewcommand{\imath}{\mathrm{i}}
\DeclareMathOperator{\Id}{Id}

\newcounter{meq}

\font\Sets=msbm10

\usepackage{enumitem}
\newcounter{classcount}

\newcounter{listcount}
\renewcommand*\thelistcount{\arabic{listcount}}

\newcounter{eqlist}

\newcounter{eqlistI}
\renewcommand*\theeqlistI{E.\arabic{eqlistI}}
\newcounter{eqlistII}
\renewcommand*\theeqlistII{E.\arabic{eqlistII}\ensuremath{^\prime}}

\theoremstyle{plain}
\newtheorem{theorem}{Theorem}

\theoremstyle{definition}

\theoremstyle{remark}
\newtheorem{remark}{Remark}

\allowdisplaybreaks

\makeatletter 
\renewcommand{\pdv}[2]{\begingroup 
\@tempswafalse\toks@={}\count@=\z@ 
\@for\next:=#2\do 
{\expandafter\check@var\next\@nil
 \advance\count@\der@exp 
 \if@tempswa 
   \toks@=\expandafter{\the\toks@\,}% 
 \else 
   \@tempswatrue 
 \fi 
 \toks@=\expandafter{\the\expandafter\toks@\expandafter\partial\der@var}}% 
\frac{\partial\ifnum\count@=\@ne\else^{\number\count@}\fi#1}{\the\toks@}% 
\endgroup} 
\def\check@var{\@ifstar{\mult@var}{\one@var}} 
\def\mult@var#1#2\@nil{\def\der@var{#2^{#1}}\def\der@exp{#1}} 
\def\one@var#1\@nil{\def\der@var{#1}\chardef\der@exp\@ne} 
\makeatother

\title{Integrable discrete autonomous quad-equations admitting, as generalized 
symmetries, known five-point differential-difference equations}
\author{R. N. Garifullin\textsuperscript{1} \and G. Gubbiotti\textsuperscript{2} \and R. I. Yamilov\textsuperscript{1}\\
}

\date{
\textsuperscript{1}Institute of Mathematics, Ufa Federal Research Centre,\\
Russian Academy of Sciences,\\ 
112 Chernyshevsky Street, Ufa 450008, Russian Federation\\
E-mail: \texttt{rustem@matem.anrb.ru}, \texttt{RvlYamilov@matem.anrb.ru}\\
\textsuperscript{2}School of Mathematics and Statistics, F07, The University of Sydney,
\\
New South Wales 2006, Australia
\\
E-mail: \texttt{giorgio.gubbiotti@sydney.edu.au}
}

\begin{document}

\maketitle

\begin{abstract}
    In this paper we construct the autonomous quad-equations which admit
    as symmetries the five-point differential-difference equations
    belonging to known lists found by Garifullin, Yamilov and Levi.
    The obtained equations are classified up to autonomous point
    transformations and some simple non-autonomous transformations.
    We discuss our results in the framework of the known literature.
    There are among them a few new examples of both sine-Gordon and Liouville type
    equations.
\end{abstract}

\section{Introduction}
\label{sec:intro}

In this paper we consider in discrete \emph{quad-equations}.
Quad-equations are four-point relations of the form:
\begin{equation}
    F \left( u_{n,m},u_{n+1,m},u_{n,m+1},u_{n+1,m+1} \right)=0,
    \label{eq:quadgen}
\end{equation}
for an unknown field $u_{n,m}$ depending on two discrete variables,
i.e. $\left( n,m \right)\in\Z^{2}$.
The function $F=F(x,y,z,w)$ is assumed to be an 
\emph{irreducible multi-affine polynomial} \cite{ABS2003}.
We recall that a polynomial is said to be multi-affine if it is
affine with respect to all its variable, and it is said to be irreducible if
it only admits trivial factors, i.e. complex constants.
Moreover, by omitting the subscript $n,m$ we underline that it does
not depend explicitly on the discrete variables, i.e. the equation
\eqref{eq:quadgen} is assumed to be \emph{autonomous}.

It is known \cite{Garifullin2012,LeviPetreraScimiterna2008,LeviYamilov2009,LeviYamilov2011,MikhailovXenitidis2013}
that the \emph{generalized symmetries} of quad-equations are given
by differential-difference equations of the form:
\begin{subequations}
    \begin{align}
        \dv{u_{n,m}}{t_{1}} &= \varphi_{n,m}\left( u_{n+k_{1},m},\dots,u_{n-k'_{1},m} \right),
        \label{eq:symm1}
        \\
        \dv{u_{n,m}}{t_{2}} &= \psi_{n,m} \left( u_{n,m+k_{2}},\dots,u_{n,m-k'_{2}} \right),
        \label{eq:symm2}
    \end{align}
    \label{eq:symmquad}
\end{subequations}
where $k_{i},k'_{i}\in\N$.
This means that one generalized symmetries depends only on shifts in the $n$
direction, while the other depends only on shifts in the $m$ direction.
We note that usually differential-difference equations are also addressed
by the amount of points of the lattice that they are involving.
In this sense the equations in \eqref{eq:symmquad} are a $k_{1}+k'_{1}+1$-point
differential-difference equation and a $k_{2}+k'_{2}+1$-point differential-difference
equation respectively.

The vast majority of quad-equations known in literature admits as generalized
symmetries three-point differential-difference equations in both directions
\cite{LeviPetreraScimiterna2008,LeviYamilov2011,Xenitidis2009proc,Xenitidis2009,GSL_symmetries,Hydon2007}.
When the given quad-equation admits three-point generalized symmetries
it can be interpreted as a B\"acklund transformation for these 
differential-difference equantions \cite{Levi1981,LeviPetreraScimiterna2008}.
More recently, several examples with more complicated symmetry structure
have been discovered \cite{Adler2011,GarifullinYamilov2012,MikhailovXenitidis2013,ScimiternaHayLevi2014,ghy15}.

In the aforementioned paper the problem which was addressed
was to find the generalized symmetries of a given quad-equation.
If the generalized symmetries are given by integrable differential-difference
equations, then one can construct a whole hierarchy of generalized symmetries
for the given quad-equation which is integrable according to the generalized symmetries
method \cite{Yamilov2006}.
On the other hand, also the converse problem can be considered: fixed a
differential-difference equation find the quad-equation admitting it
as a generalized symmetry.
This point of view was taken in \cite{LeviYamilov2009},
using one well-known five-point autonomous differential-difference
equation.
In \cite{GarifullinYamilov2015} a similar analysis was carried out
for a known class of integrable non-autonomous Volterra and Toda three-point
differential-difference equations \cite{LeviYamilov1997}.

In this paper we are going to generalize the results of 
\cite{LeviYamilov2009, GarifullinYamilov2015}:
we will start from a known integrable \emph{autonomous five-point}
differential-differential equation and construct the quad-equation 
admitting it as generalized symmetry in the $n$ direction, i.e.
of the form \eqref{eq:symm1}.
To this end we will use the classification of the five-point
differential-difference equations given in 
\cite{GarifullinYamilovLevi2016,GarifullinYamilovLevi2018}.
The equations belonging to the classification in
\cite{GarifullinYamilovLevi2016,GarifullinYamilovLevi2018}
have the following form:
\begin{equation}
    \begin{aligned}
        \dv{u_{k}}{t} &= A\left( u_{k+1},u_{k},u_{k-1} \right)u_{k+2}
    +B\left( u_{k+1},u_{k},u_{k-1} \right)u_{k-2}
    \\
    &+C\left( u_{k+1},u_{k},u_{k-1} \right),
    \end{aligned}
    \label{eq:4thgen}
\end{equation}
i.e. are linear in $u_{k\pm2}$.
These equations are divided in two classes, which we now describe explicitly.
The first class, Class I was presented in \cite{GarifullinYamilovLevi2016}
and contains all the equations of the form \eqref{eq:4thgen} such
that the condition:
\begin{equation}
   A\neq \alpha\left( u_{k+1},u_{k} \right)\alpha\left( u_{k},u_{k-1} \right),
   \quad
   B\neq \beta\left( u_{k+1},u_{k} \right)\left( u_{k},u_{k-1} \right),
   \label{eq:condAB}
\end{equation}
holds true for any functions $\alpha,\ \beta$. This Class contains seventeen equations.
Equations belonging to this class are denoted by (E.x),
where x is an Arabic number.
The second class, Class II, was presented in \cite{GarifullinYamilovLevi2018}
and contains all the equations of the form \eqref{eq:4thgen} such
that the condition \eqref{eq:condAB} does not hold, i.e. there exist functions $\alpha,\ \beta$ such that
\begin{equation}
   A= \alpha\left( u_{k+1},u_{k} \right)\alpha\left( u_{k},u_{k-1} \right),
   \quad
   B= \beta\left( u_{k+1},u_{k} \right)\left( u_{k},u_{k-1} \right).
   \label{eq:condAB2}
\end{equation}
This class contains fourteen equations.
Equations belonging to this class are denoted by (E.x\ensuremath{^\prime}),
where x is an Arabic number.
For easier understanding of the results, we split the complete list
into smaller Lists 1-6. In each List the equations are related to each
other by autonomous non-invertible and non-point transformations or simple non-autonomous point transformations.
Moreover, for sake of simplicity, since the equations are autonomous, in
displaying the equations we will use the shorthand notation
$u_{i}=u_{k+i}$.

\setcounter{eqlistI}{0}
\setcounter{eqlistII}{0}
\begin{description}[%
  before={\setcounter{listcount}{0}},%
  ,font=\bfseries\stepcounter{listcount}List \thelistcount.~]
  \item Equations related to the double Volterra equation:
       \begin{align}
            \dv{u_0}{t}&=u_0(u_{2}-u_{-2}),
            \stepcounter{eqlistI}
            \tag{\theeqlistI}
            \label{Vol} 
            \\
            \dv{u_0}{t}&=u_0^2(u_{2}-u_{-2}),
            \stepcounter{eqlistI}
            \tag{\theeqlistI}
            \label{Vol0}
            \\
            \dv{u_0}{t}&=(u_0^2+u_0)(u_{2}-u_{-2}),
            \stepcounter{eqlistI}
            \tag{\theeqlistI}
            \label{Vol1}
            \\
            \dv{u_0}{t}&=(u_{2}+u_{1})(u_0+u_{-1})-(u_1+u_0)(u_{-1}+u_{-2}),
            \stepcounter{eqlistI}
            \tag{\theeqlistI}
            \label{Vol_mod}
            \\
            \dv{u_0}{t}&
            \begin{aligned}[t]
                &=(u_{2}-u_{1}+a)(u_0-u_{-1}+a)
                \\
                &+(u_1-u_0+a)(u_{-1}-u_{-2}+a)+b,
            \end{aligned}
            \stepcounter{eqlistI}
            \tag{\theeqlistI}
            \label{Vol_mod1}
            \\
            \dv{u_0}{t}&
            \begin{aligned}[t]
                &=u_{2}u_{1}u_0(u_0u_{-1}+1)\\
                &-(u_1u_0+1)u_0u_{-1}u_{-2}+u_0^2(u_{-1}-u_{1}),
            \end{aligned}
            \stepcounter{eqlistI}
            \tag{\theeqlistI}
            \label{Vol2}
            \\
           \dv{u_0}{t}&=u_{{0}} \left[u_1(u_2-u_0)+u_{-1}(u_0-u_{-2}) \right],
            \stepcounter{eqlistII}
            \tag{\theeqlistII}
            \label{eq1ii}
            \\
            \dv{u_0}{t}&=u_{{1}}{u_{{0}}}^{2}u_{-1} \left( u_{{2}}-u_{{-2}} \right).
            \stepcounter{eqlistII}
            \tag{\theeqlistII}
            \label{eq2ii}
       \end{align}
        Transformations $\tilde u_k=u_{2k}$ or $\tilde u_k=u_{2k+1}$ turn equations 
        \eqref{Vol}-\eqref{Vol1} into the well-known Volterra equation and its modifications 
    in their standard form.
       The other equations are related to the \emph{double
       Volterra equation} \eqref{Vol} through some autonomous 
        non-invertible non-point transformations. 
        We note that equation \eqref{eq2ii} was presented in 
        \cite{AdlerPostnikov2008}.
   \item Linearizable equations:
       \begin{align}
           \dv{u_0}{t}&
           \begin{aligned}[t]
               &=(T-a)
           \left[\frac{(u_1+au_0+b)(u_{-1}+au_{-2}+b)}{u_{0}+au_{-1}+b}+u_0+au_{-1}+b\right]
           \\
           &+cu_0+d,
           \end{aligned}
            \stepcounter{eqlistI}
            \tag{\theeqlistI}
            \label{Bur2}
            \\
            \dv{u_0}{t}&=\frac{u_2u_0}{u_1}+u_1-a^2\left(u_{-1}+\frac{u_0u_{-2}}{u_{-1}}\right)+cu_0.
            \stepcounter{eqlistI}
            \tag{\theeqlistI}
            \label{Bur}
       \end{align}
       In both equations $ a\neq0,$ in \eqref{Bur2} $(a+1)d=bc$, and $T$ is the translation
       operator $Tf_{n}=f_{n+1}$.

       Both equations of List \thelistcount\ are related to the linear equation:
       \begin{equation}
       \label{lin_eq} 
       \dv{u_0}{t}=u_2-a^2 u_{-2}+\frac{c}{2}u_0 
        \end{equation}
        through an autonomous non-invertible non-point transformations.
        We note that \eqref{Bur2} is linked to \eqref{lin_eq}
        with a transformation which is implicit in both directions,
        see \cite{GarifullinYamilovLevi2016} for more details.
    \item Equations related to a generalized symmetry of the  Volterra equation:
        \begin{align}
            \dv{u_0}{t}&
            \begin{aligned}[t]
                &=u_{{0}} 
                \left[u_1(u_2+u_1+u_0)-u_{-1}(u_0+u_{-1}+u_{-2}) \right]
                \\
                &+cu_{{0}} \left( u_{{1}}-u_{{-1}} \right),
            \end{aligned}
            \stepcounter{eqlistII}
            \tag{\theeqlistII}
            \label{Vol1s}
            \\
            \dv{u_0}{t} &
            \begin{aligned}[t]
                &= (u_{{0}}^2-a^2) 
            \left[(u_1^2-a^2)(u_2+u_0)-(u_{-1}^2-a^2)(u_0+u_{-2}) \right]
            \\
            &+c(u_{{0}}^2-a^2) \left( u_{{1}}-u_{{-1}} \right),
            \end{aligned}
            \stepcounter{eqlistII}
            \tag{\theeqlistII}
            \label{mVol2}
            \\
            \dv{u_0}{t}&=(u_1-u_0+a)(u_0-u_{-1}+a)(u_2-u_{-2}+4a+c)+b,
            \stepcounter{eqlistII}
            \tag{\theeqlistII}
            \label{Volz}
            \\
            \dv{u_0}{t}&=u_0[u_1(u_2-u_1+u_0)-u_{-1}(u_0-u_{-1}+u_{-2})],
            \stepcounter{eqlistII}
            \tag{\theeqlistII}
            \label{Vol_mod1s}
            \\
            \dv{u_0}{t}&=(u_{{0}}^2-a^2) \left[(u_1^2-a^2)(u_2-u_0)+(u_{-1}^2-a^2)(u_0-u_{-2}) \right],
            \stepcounter{eqlistII}
            \tag{\theeqlistII}
            \label{mVol3}
            \\
            \dv{u_0}{t}&=(u_1+u_0)(u_0+u_{-1})(u_2-u_{-2}).
            \stepcounter{eqlistII}
            \tag{\theeqlistII}
            \label{Vol_mod2}
        \end{align}
        These equations are related between themselves by some transformations, 
        for more details see \cite{GarifullinYamilovLevi2018}.
        Moreover equations (\ref{Vol1s},\ref{mVol2},\ref{Volz}) are the generalized symmetries of some known
        three-point autonomous 	differential-difference equations \cite{Yamilov2006}.
    \item  Equations of the relativistic Toda type:
        \begin{align}
            \dv{u_0}{t}&=(u_0-1)\left(\frac{u_2(u_1-1)u_0}{u_1}-\frac{u_0(u_{-1}-1)u_{-2}}{u_{-1}}-u_1+u_{-1}\right),
            \stepcounter{eqlistI}
            \tag{\theeqlistI}
        \label{our1}
        \\
        \dv{u_0}{t}&
        \begin{aligned}[t]
        &=\frac{u_2u_1^2u_0^2(u_0u_{-1}+1)}{u_1u_0+1}-\frac{(u_1u_0+1)u_0^2u_{-1}^2u_{-2}}{u_0u_{-1}+1}
        \\
        &-\frac{(u_1-u_{-1})(2u_1u_0u_{-1}+u_1+u_{-1})u_0^3}{(u_1u_0+1)(u_0u_{-1}+1)},
        \end{aligned}
            \stepcounter{eqlistI}
            \tag{\theeqlistI}
        \label{our2}
        \\
            \dv{u_0}{t}&=(u_1u_0-1)(u_0u_{-1}-1)(u_2-u_{-2}).
            \setcounter{eqlistII}{13}
            %\stepcounter{eqlistII}
            \tag{\theeqlistII} 
            \label{ourii}
        \end{align}
        Equation \eqref{ourii} was known 
        \cite{GarifullinYamilov2012,GarifullinMikhailovYamilov2014} to be
        is a relativistic Toda type equation.
        Since in \cite{GarifullinYamilovLevi2016} it was shown that the equations of List 
        \thelistcount\ are related through autonomous non-invertible non-point transformations,
        it was suggested that \eqref{our1} and \eqref{our2} should be of the same type.
        Finally, we note that equation \eqref{our1} appeared in \cite{Adler2016b} earlier than in 
        \cite{GarifullinYamilovLevi2016}.
        \setcounter{eqlistII}{8}
    \item Equations related to the Itoh-Narita-Bogoyavlensky (INB) equation:
        \begin{align}
            \label{INB}
            \dv{u_0}{t}&=u_0(u_2+u_1-u_{-1}-u_{-2}),
            \stepcounter{eqlistI}
            \tag{\theeqlistI}
            \\
            \stepcounter{eqlistI}
            \tag{\theeqlistI}
            \label{mod_INB}
            \dv{u_0}{t}&
            \begin{aligned}[t]
                &=(u_{2}-u_{1}+a)(u_0-u_{-1}+a)
                \\
                &+(u_1-u_0+a)(u_{-1}-u_{-2}+a)
                \\
                &+(u_1-u_0+a)(u_0-u_{-1}+a)+b,
            \end{aligned}
            \\
            \dv{u_0}{t} &= (u_0^2+au_0)(u_2u_1-u_{-1}u_{-2}),
            \stepcounter{eqlistI}
            \tag{\theeqlistI}
            \label{mikh1}
            \\
            \stepcounter{eqlistI}
            \tag{\theeqlistI}
            \label{eq1}
            \dv{u_0}{t} &= (u_1-u_0)(u_0-u_{-1})\left(\frac {u_2}{u_1}-\frac{u_{-2}}{u_{-1}}\right),
            \\
            \dv{u_0}{t}&=u_0(u_2u_1-u_{-1}u_{-2}),
            \stepcounter{eqlistII}
            \tag{\theeqlistII}
            \label{INB1}
            \\
            \dv{u_0}{t}&=(u_1-u_0+a)(u_0-u_{-1}+a)(u_2-u_1+u_{-1}-u_{-2}+2a)+b,
            \stepcounter{eqlistII}
            \tag{\theeqlistII}
            \label{INB3}
            \\
            \dv{u_0}{t}&=u_0(u_1u_0-a)(u_0u_{-1}-a)(u_2u_1-u_{-1}u_{-2}),
            \stepcounter{eqlistII}
            \tag{\theeqlistII}
            \label{MX2}
            \\
            \dv{u_0}{t}&=(u_1+u_0)(u_0+u_{-1})(u_2+u_1-u_{-1}-u_{-2}).
            \stepcounter{eqlistII}
            \tag{\theeqlistII}
            \label{INB2}
        \end{align}
        Equation \eqref{INB} is the well-known INB equation 
        \cite{Bogoyavlensky1988,Itoh1975,Narita1982}.
        Equations \eqref{mod_INB} with $a=0$ and \eqref{mikh1} with $a=0$ 
        are  simple modifications of the INB and were presented  in \cite{MikhailovXenitidis2013} 
        and \cite{Bogoyavlensky1991}, respectively.
        Equation \eqref{mikh1} with $a=1$ has been found in 
        \cite{AdlerPostnikov2014,MikhailovXenitidis2013}.
        Equation \eqref{INB1} is a well-known modification of  INB equation \eqref{INB}, 
        found by Bogoyalavlesky himself \cite{Bogoyavlensky1988}. 
        Finally, equation \eqref{MX2} with $a=0$ was considered in \cite{AdlerPostnikov2008}. 
        All the equations in this list can be reduced to the INB equation
        using autonomous non-invertible non-point transformations.
        Moreover, equations \eqref{mod_INB},\eqref{eq1} and \eqref{INB1} are related through
        non-invertible transformations to the equation:
        \begin{align}
            \dv{u_{0}}{t} &= \left( u_{2}-u_{0} \right)\left( u_{1}-u_{-1} \right)
                \left( u_{0}-u_{-2} \right).
            \label{eq311}
        \end{align}
        We will also study  below equation \eqref{eq311}, see Remark \ref{rem_R} in Section \ref{sec:method}.
    \item Other equations:
        \begin{align}
            \dv{u_0}{t} &= u_0^2(u_2u_1-u_{-1}u_{-2})-u_0(u_1-u_{-1}),
            \stepcounter{eqlistI}
            \tag{\theeqlistI}
            \label{seva}
            \\
            \dv{u_0}{t} &= (u_0+1)
            \left[\frac{u_2u_0(u_{1}+1)^2}{u_1}-\frac{u_{-2}u_0(u_{-1}+1)^2}{u_{-1}}
            +(1+2u_0)(u_1-u_{-1})\right],
            \stepcounter{eqlistI}
            \tag{\theeqlistI}
            \label{rat}
            \\
            \dv{u_0}{t} &= (u_0^2+1)\left(u_2\sqrt{u_1^2+1}-u_{-2}\sqrt{u_{-1}^2+1}\right),
            \stepcounter{eqlistI}
            \tag{\theeqlistI}
            \label{sroot}
            \\
            \dv{u_0}{t}&=u_1u_0^3u_{-1}(u_2u_1-u_{-1}u_{-2})-u_0^2(u_1-u_{-1}).
            \setcounter{eqlistII}{14}
            \tag{\theeqlistII}
            \label{SK2}
        \end{align}
        Equation \eqref{seva} has been found in \cite{TsujimotoHirota1996} 
        and it is called the discrete Sawada-Kotera equation 
        \cite{Adler2011,TsujimotoHirota1996}. 
        Equation \eqref{SK2} is a simple modification of the discrete
        Sawada-Kotera equation \eqref{seva}. 
                Equation \eqref{rat} has been found in \cite{Adler2016b} and 
        is related to \eqref{seva}.
        On the other hand equation \eqref{sroot} has been found as 
        a result of the classification in \cite{GarifullinYamilovLevi2016} 
        and seems to be a new equation. 
        It was shown in \cite{GarifullinYamilov2017} that equation 
        \eqref{sroot} is a discrete analogue of the Kaup-Kupershmidt
        equation.
        Then we will refer to equation \eqref{sroot} as the discrete
        Kaup-Kupershmidt equation.
        No transformation into known equations 
        of equation \eqref{sroot} is known.
\end{description}

As a result we will obtain several examples of autonomous nontrivial 
quad-equations.
By nontrivial equations we mean non-degenerate, irreducible
and nonlinear equations. Moreover we consider as trivial also the
equations
\begin{subequations}
    \begin{align}
        u_{n+1,m+1}u_{n,m+1}=\kappa_1 u_{n+1,m}u_{n,m},
        \label{ple1}
        \\ 
        u_{n+1,m+1}u_{n,m}=\kappa_2u_{n+1,m}u_{n,m+1},
        \label{ple2}
    \end{align}
    \label{ple}
\end{subequations}
which are equavalent to a linear one up to the transcendental transformation 
$u_{n,m}=\exp v_{n,m}$.
The resulting equations belong to two different types of equations:
they can be either \emph{Liouville type} LT or \emph{sine-Gordon type}
sGT.
Liouville type equations are quad-equation which are \emph{Darboux integrable},
i.e. they possess 
\emph{first integrals}, one containing only shifts in the first
direction and the other containing only shifts in the second direction.
This means that there exist two functions:
\begin{subequations}
    \begin{align}
        W_1=W_{1,n,m}(u_{n+l_1,m},u_{n+l_1+1,m},\ldots,u_{n+k_1,m}),
        \label{eq:darbfint1}
        \\
        W_2=W_{2,n,m}(u_{n,m+l_2},u_{n,m+l_2+1},\ldots,u_{n,m+k_2}),
        \label{eq:darbfint2}
    \end{align}
    \label{eq:darbfint}
\end{subequations}
where $l_1<k_1$ and $l_2<k_2$ are integers, such that the relations
\begin{subequations}
    \begin{align}
        (T_n-\Id)W_2=0, 
        \label{eq:darb1}
        \\
        (T_m-\Id)W_1=0
        \label{eq:darb2}
    \end{align}
    \label{eq:darbdef}
\end{subequations}
hold true identically on the solutions of 
\eqref{eq:quadgen}.
By $T_n,T_m$ we denote the shift operators 
in the first and second 
directions, i.e. $T_n h_{n,m}=h_{n+1,m}$, $T_m h_{n,m}=h_{n,m+1}$,
and by $\Id$ we denote the identity 
operator $\Id h_{n,m}=h_{n,m}$.
The number $k_{i}-l_{i}$, where
$i=1,2$, is called the \emph{order} of
the first integral $W_{i}$.
On the other hand a quad-equation is said to be of sine-Gordon type
if it is not Darboux integrable and possess two generalized symmetries
of the form \eqref{eq:symmquad}.
Such equations are usually integrable by the inverse scatteting 
method.

We underline that, by construction, all the equations we are going
to produce will possess a generalized symmetry in the $n$ direction \eqref{eq:symm1}
which is not sufficient for integrability.
To prove that the equation, we are going to produce, are LT or sGT
we will either identify them with known equations, or present
their first integrals in the sense of \eqref{eq:darbfint} or a 
generalized symmetry in the $m$ direction \eqref{eq:symm2}.

The plan of the paper is then the following:
In Section \ref{sec:method} we discuss the theoretical background
which allows us to make the relevant computations.
In Section \ref{sec:results} we enumerate all the quad-equations
of the form \eqref{eq:quadgen} which corresponds to the equations
of the Class I and II and describe their properties.
Finally in Section \ref{sec:summary} we give a summary of the work
and some outlook for future works in the field.

\section {The method}
\label{sec:method}

The most general autonomous multi-affine quad-equation has the following
form:
\begin{equation}
    \begin{aligned}
        F\equiv u_{{n+1,m+1}}u_{{n,m+1}} 
        &\left( a_{{1}}u_{{n,m}}u_{{n+1,m}}+a_{{2}}u_{{n,m}}+a_{{3}}u_{{n+1,m}}+a_{{4}} \right) 
        \\
        +u_{{n+1,m+1}} 
        &\left( b_{{1}}u_{{n,m}}u_{{n+1,m}}+b_{{2}}u_{{n,m}}+b_{{3}}u_{{n+1,m}}+b_{{4}} \right) 
        \\
        +u_{{n,m+1}} &\left( c_{{1}}u_{{n,m}}u_{{n+1,m}}+c_{{2}}u_{{n,m}}+c_{{3}}u_{{n+1,m}}+c_{{4}} \right) 
        \\
        +&d_{{1}}u_{{n,m}}u_{{n+1,m}}+d_{{2}}u_{{n,m}}+d_{{3}}u_{{n+1,m}}+d_{{4}}=0.
    \end{aligned}
    \label{F}
\end{equation}
From \eqref{eq:symm1} and \eqref{eq:4thgen} we have that
the most general symmetry in the $n$ direction belonging to
Class I and II can be written as:
\begin{equation}
    \begin{aligned}
        \dot u_{n,m} &= a(u_{n+1,m}, u_{n,m}, u_{n-1,m}) u_{n+2,m}+b(u_{n+1,m}, u_{n,m}, u_{n-1,m}) u_{n-2,m}
        \\
        &+ c(u_{n+1,m}, u_{n,m}, u_{n-1,m}).
    \end{aligned}
    \label{e1}
\end{equation}

We have then the following general result, analogous to Theorem 2 
in \cite{GarifullinYamilov2015}:

\begin{theorem}
    If the autonomous quad-equation \eqref{F} admits a generalized
    symmetry of the form \eqref{e1} then it has the following form:
    \begin{equation}
        \begin{aligned}
            \hat{F}&\equiv \alpha u_{{n+1,m+1}}u_{{n,m+1}}
            +u_{{n+1,m+1}} \left( \beta_{{1}}u_{{n,m}}+\beta_{{2}} \right) 
            +u_{{n,m+1}} \left(\gamma_{{1}}u_{{n+1,m}}+\gamma_{{2}} \right) 
            \\
            &+\delta_{{1}}u_{{n,m}}u_{{n+1,m}}+\delta_{{2}}u_{{n,m}}
            +\delta_{{3}}u_{{n+1,m}}+\delta_{{4}}=0.
        \end{aligned}
        \label{Fh}
    \end{equation}
    \label{thm:struct}
\end{theorem}

\begin{proof}
    By multi-linearity we can always solve \eqref{F} with respect to
    $u_{n+1,m+1}$ and write:
    \begin{equation}
        u_{n+1,m+1}=f(u_{n+1,m},u_{n,m},u_{n,m+1}).
    \end{equation}
    From Theorem 2 in \cite{LeviYamilov2009} we have that a quad-equation
    admits a generalized symmetry of the form \eqref{eq:symm1} with
    $k_{1}=-k'_{1}=2$ if the following conditions are satisfied:
    \begin{subequations}  
        \begin{align}
            (T_n-T_n^{-1})\log \dfrac{f}{u_{n+1,m}}&=(1-T_{m})\log \pdv{\varphi}{u_{n+2,m}},
            %a(u_{n+1,m}, u_{n,m}, u_{n-1,m}),
            \label{(33)}
            \\
            (T_n^{-2}-1)\log (\dfrac{f}{u_{n,m}}\middle/\dfrac{f}{u_{n,m+1}})&=(1-T_{m})\log \pdv{\varphi}{u_{n-2,m}}.
            %b(u_{n+1,m}, u_{n,m}, u_{n-1,m}).
            \label{(34)}
        \end{align}
        \label{eq:3334}
    \end{subequations}
    In \eqref{eq:3334} we suppressed the indices $n,m$ in $\varphi$ since we
    are dealing with autonomous differential-difference equations.
    Now from \eqref{e1} we have:
    \begin{equation}
        \pdv{\varphi}{u_{n+2,m}} = a(u_{n+1,m}, u_{n,m}, u_{n-1,m}),
        \quad
        \pdv{\varphi}{u_{n-2,m}}=b(u_{n+1,m}, u_{n,m}, u_{n-1,m}).
        \label{eq:e1diff}
    \end{equation}
    From equation \eqref{eq:e1diff} we have differentiating \eqref{(33)} 
    w.r.t. $u_{n+2,m}$:
    \begin{equation}
        0=\dfrac{}{u_{n+2,m}}T_n\log\dfrac{f}{u_{n+1,m}}=
        -2T_n\left(\frac{\partial^{2}F/\partial u_{n+1,m+1}\partial u_{n+1,m}}{\partial F/\partial{u_{n+1,m+1}}}\right),
        \label{eq:cond1}
    \end{equation}
    where we used the implicit function theorem and the fact that $F$
    is multi-affine \eqref{F}.
    In the same way differentiating \eqref{(34)} w.r.t $u_{n-2,m}$
    we get:
    \begin{equation}
        0=\dfrac{}{u_{n-2,m}}T^{-2}_n\log(\dfrac{f}{u_{n,m}}\middle/\dfrac{f}{u_{n,m+1}})=
        -2T_n^{-2}\left(\frac{\partial^{2}F/\partial u_{n,m+1}\partial u_{n,m}}{\partial F/\partial{u_{n,m+1}}}\right).
        \label{eq:cond2}
    \end{equation}
    Therefore we get the two conditions:
    \begin{equation}
        \dfrac{^2F}{u_{n+1,m+1}\partial u_{n+1,m}}=\dfrac{^2F}{u_{n,m+1}\partial u_{n,m}}=0. 
        \label{eq:cond3}
    \end{equation}
    Working out explicitly the conditions in \eqref{eq:cond3} and using
    \eqref{F} we obtain:
    \begin{equation}
        a_1=a_2=a_3=b_1=b_3=c_1=c_2=0.
        \label{rest}
    \end{equation}
    Relabeling the parameters as
    \begin{equation}
        a_{4}\to \alpha,
        \quad
        b_{2}\to \beta_{1},
        \quad
        b_{4}\to \beta_{2},
        \quad
        c_{3}\to \gamma_{1},
        \quad
        c_{4}\to \gamma_{2},
        \quad
        d_{i} \to \delta_{i},
        \label{eq:relabel}
    \end{equation}
    equation \eqref{Fh} follows.
\end{proof}

\begin{remark}
    We note that the condition \eqref{eq:cond3} is the same as formula
    (29) in \cite{GarifullinYamilov2015} even though the class
    of considered differential-difference equation is different.
    \label{rem:cond29}
\end{remark}

\begin{remark}\label{rem_R}
   It can be proved in a similar way that Theorem \ref{thm:struct} is valid for generalized
    symmetries of the form \eqref{eq:symm1} satisfying the conditions $k_1'=k_1$ and 
    \begin{equation}
        \dfrac{^2\phi_{n,m}}{u_{n+k_1}^2}=0,\quad \dfrac{^2\phi_{n,m}}{u_{n-k_1}^2}=0.
        \label{eq:gencondsymm}
    \end{equation}
\end{remark}

Theorem \ref{thm:struct} tells us that the most general form of
a quad-equation admitting a five-point generalized symmetry in the $n$
direction of the form \eqref{e1} is given by \eqref{Fh}.
At this point we can pick up any of the members of Class I and II
and follow the scheme of \cite[App. B]{GarifullinYamilovLevi2016},
i.e. we fix a specific form of $a$, $b$ and $c$ in \eqref{e1}.
We beging by imposing the exponential integrability conditions 
\eqref{(33)} and \eqref{(34)} and finally we impose the symmetry condition:
\begin{equation}
    \sum_{i,j\in\left\{ 0,1 \right\}} 
    \varphi\left( u_{n+2+i,m+j},\ldots,u_{n-2+i,m+j} \right)
    \pdv{\hat{F}}{u_{n+i,m+j}}=0
    \label{eq:symmcond}
\end{equation}
which must be satisfied on all the solutions of \eqref{Fh}.
Then we express in \eqref{eq:symmcond} the functions
$u_{n+2,m+1}$, $u_{n+1,m+1}$ and $u_{n-1,m+1}$ in terms of the
independent variables
\begin{equation}
    u_{n+2,m}, \, u_{n+1,m}, \, u_{n,m}, \, u_{n-1,m}, \, u_{n,m+1}.
    \label{eq:indvar}
\end{equation}
Taking the numerator of \eqref{eq:symmcond} we obtain a polynomial
in the independent variable \eqref{eq:indvar}. This polynomial must
be identically zero, so we can equate to zero all its coefficients.

This lead us to a system of algebraic equation in the coefficients
of the multi-affine function \eqref{Fh}.
Doing so we reduce the problem of finding a quad-equation admitting
a given five-point symmetry of the form \eqref{e1} to the problem of 
solving system of algebraic equations.
Such system can be solved using a Computer Algebra System like Maple,
Mathematica or Reduce.
Amongst the possible solutions of the system we choose the non-degenerate
ones.
The non-degeneracy condition is the following one:
\begin{equation}
    \pdv{F}{u_{n+1,m+1}}\pdv{f}{u_{n+1,m}}\pdv{f}{u_{n,m+1}}\pdv{f}{u_{n,m}}\neq0.
    \label{eq:nondeg}
\end{equation}
Note that this non-degeneracy condition includes the requirement
that a quad-equation must be given by an irreducible multi-affine
polynomial, see the introduction.

\begin{remark}
    Several equations in Class I and II, e.g. \eqref{Vol_mod1}
    or \eqref{Volz}, depend on some parameters.
    Depending on the value of the parameters there can be, in principle,
    different quad-equations admitting the given differential-difference
    equation as five-point generalized symmetries.
    When possible, in order to avoid ambiguities and simplify the problem, 
    we use some simple autonomous transformations to fix the values of some parameters.
    The remaining free parameters are then treated as unknown coefficients
    in the system of algebraic equations.
    We will describe these subcases when needed in the next section.
    \label{rem:params}
\end{remark}

The number of the resulting equations is then reduced using some point
transformations.
These point transformations are essentially of two different kind:
non-autonomous transformations of the dependent variable $u_{n,m}$:
\begin{gather}
    u_{n,m}\rightarrow (-1)^m u_{n,m},
    \tag{T.1}
    \label{trmsm}
    \\
    u_{n,m}\rightarrow \left(\frac{1+\imath\sqrt3}2\right)^m u_{n,m},
    \tag{T.2}
    \label{trmsm3}
    \\
    u_{n,m}\rightarrow \imath^m u_{n,m}
    \tag{T.3}
    \label{trim}
\end{gather}
and mirror reflections of the lattice:
\begin{align}
    m\rightarrow 1-m, 
    \quad
    &\text{i.e.}
    \quad 
    u_{n,m} \longleftrightarrow u_{n,m+1},\quad u_{n+1,m} \longleftrightarrow u_{n+1,m+1},
    \tag{T\textsuperscript{*}\kern -0.3pc.1}
    \label{trmm}
    \\
    n\rightarrow 1-n, 
    \quad
    &\text{i.e.}
    \quad 
    u_{n,m} \longleftrightarrow u_{n+1,m},\quad u_{n,m+1} \longleftrightarrow u_{n+1,m+1}.
    \label{trnm}
    \tag{T\textsuperscript{*}\kern -0.3pc.2}
\end{align}

Now, in the next section we describe the results of this search.

\section{Results}

\label{sec:results}

In this section we describe the results of the procedure outlined
in Section \ref{sec:method}.
Specifically, as described in \ref{rem:params} we will explicit
the particular cases in which the parametric equations can be divided.

\subsection{List 1}

\paragraph{Equation \eqref{Vol}:}
To equation \eqref{Vol} correspond only degenerate quad-equations
in the sense of condition \eqref{eq:nondeg}.

\paragraph{Equation \eqref{Vol0}:}
To equation \eqref{Vol0} correspond two nontrivial quad-equations:
\begin{equation}
    u_{n+1,m+1}(a u_{n,m+1}+u_{n,m})-u_{n+1,m}(u_{n,m+1}+au_{n,m})=0,
    \quad a^{2}=-1.
    \label{l4(7_)}
\end{equation}
Up to transformation \eqref{trmsm} we can reduce these two equations
to only one, say the one with $a=-\imath$.
Equation (\ref{l4(7_)}, $a=-\imath$) is a special case of equation  (7) with $a_2=-\imath$ of 
List 4 in \cite{GarifullinYamilov2012}.
This means that this is a LT equation, its first integrals
being:
\begin{equation}
    W_{1} = u_{n+1,m}u_{n-1,m},
    \quad
    W_{2} = \imath^{n+1}\frac{u_{n,m+1}-u_{n,m}}{u_{n,m+1}+u_{n,m}}. 
    \label{eq:fintI2}
\end{equation}
See \cite{GarifullinYamilov2012} for more details.

\paragraph{Equation \eqref{Vol1}:}
To equation \eqref{Vol1} correspond only degenerate quad-equations
in the sense of condition \eqref{eq:nondeg}.

\paragraph{Equation \eqref{Vol_mod}:}
To equation \eqref{Vol_mod} correspond only degenerate quad-equations
in the sense of condition \eqref{eq:nondeg} or the linear
discrete wave equation:
\begin{equation}
    u_{n+1, m+1}+u_{n, m+1}-u_{n+1, m}-u_{n, m}=0.
    \label{eq:dwave}
\end{equation}

\paragraph{Equation \eqref{Vol_mod1}:}
Equation \eqref{Vol_mod1} depends on the parameters $a$ and $b$.
If $a\neq0$ it is possible to scale it to one through a scaling
transformation.
This amounts to consider the cases $a=1$ and $a=0$.

To equation (\ref{Vol_mod1}, $a=1$) correspond only degenerate quad-equations
in the sense of condition \eqref{eq:nondeg} or the linear
discrete wave equation \eqref{eq:dwave}.
The same holds true for equation (\ref{Vol_mod1}, $a=0$).
So to all the instances of equation \eqref{Vol_mod1} correspond only trivial
or linear equations.

\paragraph{Equation \eqref{Vol2}:}
To equation \eqref{Vol2} correspond two quad-equation.
One is the trivial exponential wave equation \eqref{ple1}, while
the other one is the LT equation:
\begin{equation}
    u_{n+1,m+1}u_{n,m+1}+u_{n+1,m}u_{n,m}+1=0.
    \label{l4(9)}
\end{equation}
Equation \eqref{l4(9)} is equation (9) with $c_4=1$ of List 4 in 
\cite{GarifullinYamilov2012}.
Its first integrals are \cite{GarifullinYamilov2012}:
\begin{equation}
    W_{1} = \left( -1 \right)^{m}\left( 2u_{n+1,m}u_{n,m}+1 \right),
    \quad
    W_{2} = \left(\frac{u_{n,m+1}}{u_{n,m-1}}\right)^{(-1)^{n}}.
    \label{eq:fintI6}
\end{equation}

\paragraph{Equation \eqref{eq1ii}:}
To equation \eqref{eq1ii} correspond only degenerate quad-equations
in the sense of condition \eqref{eq:nondeg} and the trivial
linearizable equation \eqref{ple1}.

\paragraph{Equation \eqref{eq2ii}:}
Equation \eqref{eq2ii} is very rich, since it
give raise to many different equations. We have two trivial
linearizable equations of the form \eqref{ple1}, but also
twelve nontrivial equations.
We have four equations of the form:
\begin{equation}
    \left(k_1 u_{n,m}+k_3 u_{n,m+1}\right)u_{n+1,m} +\left(k_2 u_{n,m} +k_4 u_{n,m+1}\right) u_{n+1,m+1}
    =0,
    \label{eq:II2a}
\end{equation}
with
\begin{subequations}
    \begin{gather}
        k_1 = 1,\, k_2 = -1,\, k_3 = 1,\, k_4 = 1,
        \label{eq:IIa1}
        \\
        k_1 = 1,\, k_2 = \imath,\, k_3 = -\imath,\, k_4 = -1,
        \label{eq:IIa2}
        \\
        k_1 = -1,\, k_2 = \imath,\, k_3 = -\imath,\, k_4 = 1,
        \label{eq:IIa3}
        \\
        k_1 = -1,\, k_2 = -1,\, k_3 = 1,\, k_4 = -1.
        \label{eq:IIa4}
    \end{gather}
    \label{eq:IIapars}
\end{subequations}
Then four equations of the form:
\begin{equation}
    u_{n,m} u_{n+1,m}+k_1 u_{n,m+1} u_{n+1,m}+k_2 u_{n,m+1} u_{n+1,m+1}
    =0,
    \label{eq:II2b}
\end{equation}
with
\begin{subequations}
    \begin{gather}
        k_1 = -1+\imath,\, k_2 = -\imath,
        \label{eq:IIb1}
        \\
        k_1 = 1-\imath,\, k_2 = -\imath,
        \label{eq:IIb2}
        \\
        k_1 = -1-\imath,\, k_2 = \imath,
        \label{eq:IIb3}
        \\
        k_1 = 1+\imath,\, k_2 = \imath.
        \label{eq:IIb4}
    \end{gather}
    \label{eq:IIbpars}
\end{subequations}
Finally we have four equations of the form:
\begin{equation}
    u_{n,m} u_{n+1,m}+k_1 u_{n,m} u_{n+1,m+1}+k_2 u_{n,m+1} u_{n+1,m+1}
    =0,
    \label{eq:II2c}
\end{equation}
with:
\begin{subequations}
    \begin{gather}
        k_1 = -1+\imath,\, k_2 = -\imath,
        \label{eq:IIc1}
        \\
        k_1 = 1-\imath,\, k_2 = -\imath,
        \label{eq:IIc2}
        \\
        k_1 = -1-\imath,\, k_2 = \imath,
        \label{eq:IIc3}
        \\
        k_1 = 1+\imath,\, k_2 = \imath.
        \label{eq:IIc4}
    \end{gather}
    \label{eq:IIcpars}
\end{subequations}

Up to the transformations (\ref{trmsm},\ref{trim}) and (\ref{trmm}) we 
can reduce equations (\ref{eq:II2a},\ref{eq:II2b},\ref{eq:II2c}) to only two equations,
say \eqref{eq:II2a} with the parameters \eqref{eq:IIa4} and
\eqref{eq:II2b} with the parameters \eqref{eq:IIb4}:
\begin{subequations}
    \begin{gather}
        (u_{n,m}-u_{n,m+1})u_{n+1,m}+(u_{n,m+1}+u_{n,m})u_{n+1,m+1}=0,
        \label{l4(7)}
        \\
        \left[u_{n,m}+(1+\imath)u_{n,m+1}\right]u_{n+1,m}+\imath u_{n+1,m+1}u_{n,m+1}=0.
        \label{nDar}
    \end{gather}
    \label{II2def}
\end{subequations}
For a complete description of the transformations leading to \eqref{II2def}
see Figure \ref{fig:II2transf}.

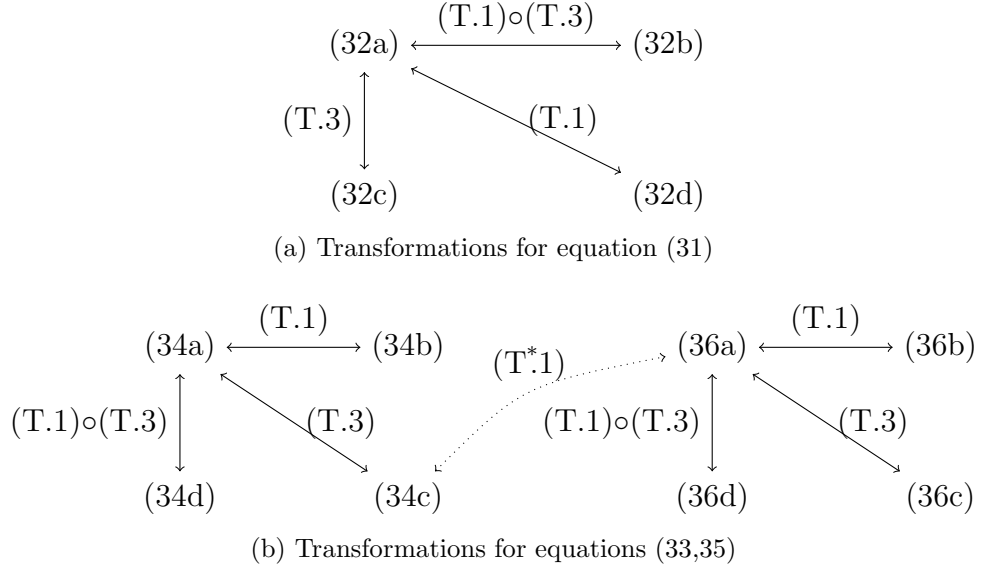
\begin{figure}[hbtp]
    \centering
    \subfloat[][Transformations for equation \eqref{eq:II2a}]{
    \begin{tikzpicture}
        \path (0,2) node(e1) {\eqref{eq:IIa1}}
            (4,2) node(e2) {\eqref{eq:IIa2}}
            (0,0) node(e3) {\eqref{eq:IIa3}}
            (4,0) node(e4) {\eqref{eq:IIa4}};
            \draw[<->] (e1) -- node[above] {\eqref{trmsm}$\circ$\eqref{trim}} (e2);
            \draw[<->] (e1) -- node[right] {\eqref{trmsm}} (e4);
            \draw[<->] (e1) -- node[left] {\eqref{trim}} (e3);
    \end{tikzpicture}
    }
    \\
    \subfloat[][Transformations for equations (\ref{eq:II2b},\ref{eq:II2c})]{
    \begin{tikzpicture}
        \path (0,2) node(e1) {\eqref{eq:IIb1}}
            (3,2) node(e2) {\eqref{eq:IIb2}}
            (3,0) node(e3) {\eqref{eq:IIb3}}
            (0,0) node(e4) {\eqref{eq:IIb4}}
            (7,2) node(e5) {\eqref{eq:IIc1}}
            (10,2) node(e6) {\eqref{eq:IIc2}}
            (10,0) node(e7) {\eqref{eq:IIc3}}
            (7,0) node(e8) {\eqref{eq:IIc4}};
            \draw[<->] (e1) -- node[above] {\eqref{trmsm}} (e2);
            \draw[<->] (e1) -- node[left] {\eqref{trmsm}$\circ$\eqref{trim}} (e4);
            \draw[<->] (e1) -- node[right] {\eqref{trim}} (e3);
            \draw[<->,dotted] (e3) ..  controls (4.5,1.5) .. node[dotted,above] {\eqref{trmm}} (e5);
            \draw[<->] (e5) -- node[above] {\eqref{trmsm}} (e6);
            \draw[<->] (e5) -- node[left] {\eqref{trmsm}$\circ$\eqref{trim}} (e8);
            \draw[<->] (e5) -- node[right] {\eqref{trim}} (e7);
    \end{tikzpicture}
    }
    \caption{The relationship between the equations (\ref{eq:II2a},\ref{eq:II2b},\ref{eq:II2c}).
    By \eqref{trmsm}$\circ$\eqref{trim} we mean the composition of the
    two transformations.}
    \label{fig:II2transf} 
\end{figure}

Equation \eqref{l4(7)} is a LT equation.
It can identified with equation (7) with $a_2=1$ of List 4 of
\cite{GarifullinYamilov2012}.
In particular its first integrals are given by:
\begin{equation}
    W_{1} = \left( -1 \right)^{m}u_{n+1,m}u_{n-1,m},
    \quad
    W_{2} = \imath^{n} \frac{u_{n,m}+\imath u_{n,m+1}}{\imath u_{n,m}+u_{n,m+1}}.
    \label{eq:fintl4(7)}
\end{equation}

Equation \eqref{nDar} is a LT equation too.
Its first integrals are given by:
\begin{equation}
    W_1=u_{n,m}u_{n+1,m}u_{n+2,m}u_{n+3,m},
    \quad 
    W_2=(-\imath)^n\frac{u_{n,m}+\imath u_{n,m+1}}{u_{n,m}+u_{n,m+1}}.
    \label{eq:fintnewdarb}
\end{equation}
Note that the $W_{2}$ first integral of \eqref{nDar} is very
similar to the $W_{2}$ first integral of \eqref{l4(7)}.
Moreover, the $W_{1}$  first integral of \eqref{nDar} is a generalization
of the $W_{1}$ first integral of \eqref{l4(7)} with more points.
It is possible to show that equation \eqref{nDar} do not belong to 
known families of Darboux integrable equations 
\cite{GarifullinYamilov2012,AdlerStartsev1999,GSY_DarbouxI,Startsev2014,Startsev2010},
because of the specific form of $W_1$.
Therefore we believe that equation \eqref{nDar} is a new Darboux
integrable equation.

\subsection{List 2}

\paragraph{Equation \eqref{Bur2}:}

Equation \eqref{Bur2} depends on four parameters $a$, $b$, $c$ and $d$
linked among themselves by the condition $\left( a+1 \right)d=bc$.
Therefore we have that using a linear transformation $u_{n,m}\to\alpha u_{n,m}+\beta$
we need to consider only three different cases: 1) $ a\neq0,a\neq -1,b=0,d=0$, 
2) $a=-1,b=1,c=0$,  3) $a=-1,b=0$ (recall that $a\neq0$ in all cases).

\begin{description}
    \item[Case $ a\neq0,a\neq -1,b=0,d=0$:]
        For these values of the parameters we obtain either
        linear equations or exponential wave equations \eqref{ple2}
        if also $a\neq1$ and $c\neq0$.
        For $a=1$ and $c=0$ we obtain the equation:
        \begin{equation}
            \begin{aligned}
                &\phantom{+}d_3^{4}u_{{n,m}}u_{{n+1,m}}+d_3^{3}d_4  \left( u_{{n+1,m}}+u_{{n,m}} \right) 
            -{  c_3} d_{3}d_4^{2} \left( u_{{n,m+1}}+u_{{n+1,m+1}}\right)
            \\
            &+d_3^{2}d_4 \left[ d_4+c_3\left( u_{{n,m}}u_{{n+1,m+1}}+u_{{n,m+1}}u_{{n+1,m}} \right)\right] 
            +c_3^{2} d_4^{2}u_{{n,m+1}}u_{{n+1,m+1}}=0.
            \end{aligned}
            \label{eq:I7a1b0c0d0}
        \end{equation}
        In equation \eqref{eq:I7a1b0c0d0} we have the condition
        $c_3d_3d_4 \neq 0$.
        Using a linear transformation of the form 
        $u_{n,m} = A \hat{u}_{n,m}+B$ we can reduce equation
        \eqref{eq:I7a1b0c0d0} to two different cases.
        \begin{description}
            \item[Case 1:] If $d_3^2+c_3d_4=0$ and if we chose
            $A=d_4/d_3=-4 B$, we have:
            \begin{equation}
                \begin{aligned}
                    \left( u_{n+1,m+1}-u_{n+1,m} \right)\left( u_{n,m}-u_{n,m+1} \right)
                    +u_{n+1,m+1}+u_{n+1,m}+u_{n,m+1}+u_{n,m}=0.
                \end{aligned}
                \label{eq:I7a1b0c0d0case1}
            \end{equation}
            This equation is equation (1) of List 3 of \cite{GarifullinYamilov2012},
            hence it is a LT equation with the following first integrals:
            \begin{subequations}
                \begin{align}
                    W_{1} & = \frac{2\left( u_{n+1,m}+u_{n,m} \right)+1}{%
                    \left( u_{n+2,m}-u_{n,m} \right)\left( u_{n+1,m}-u_{n-1,m} \right)},
                    \label{eq:I7a1b0c0d0case1int1}
                    \\
                    W_{2} & = \left( -1 \right) ^{n}
                    \frac { u_{{n,m-1}}+u_{{n,m+1}}-2 \left(u_{{n,m}}-1\right) }{u_{{n,m+1}}-u_{{n,m-1}}}.
                    \label{eq:I7a1b0c0d0case1int2}
                \end{align}
                \label{eq:I7a1b0c0d0case1int}
            \end{subequations}
            \item[Case 2:] If $d_3^2+c_3d_4\neq0$ we chose
                \begin{equation}
                    A=2\frac{d_4}{d_3}, 
                    \quad  
                    B=\frac{d_3d_4}{d_3^2+c_3d_4},
                    \label{eq:I7a1b0c0d0case2def}
                \end{equation}
                and equation \eqref{eq:I7a1b0c0d0} reduces to:
                \begin{equation}
                    \begin{aligned}
                        \left( \frac{c_3d_4}{d_3^{2}}u_{{n+1,m+1}}+u_{{n+1,m}} \right)  
                            \left(\frac{c_3d_4}{d_3^{2}}u_{{n,m+1}}+u_{{n,m}} \right) 
                            +u_{{n+1,m}}+u_{{n,m}}
                        +\frac{d_{3}^{2}}{d_3^2+c_3d_4}=0.
                    \end{aligned}
                    \label{eq:I7a1b0c0d0case2}
                \end{equation}
            This equation is equation (3) of List 3 of \cite{GarifullinYamilov2012}
            with $a_{3}=c_3d_4/d_3^{2}$, hence it is a LT equation with the following first integrals:
            \begin{subequations}
                \begin{align}
                    W_{1}&=
                    \left( -a_{3} \right)^{-m}
                    \frac{\left( a_{3}+1 \right)\left( u_{n+1,m}+u_{n,m} \right)+1}{%
                    \left( u_{n+2,m}-u_{n,m} \right)\left( u_{n+1,m}-u_{n-1,m} \right)},
                    \label{eq:I7a1b0c0d0case2int1}
                    \\
                W_{2}&=
                \left[ 
                    \frac{u_{n,m-1}+a_{3}u_{n,m}+1}{\sqrt{-a_{3}}\left( u_{n,m}+a_{3}u_{n,m+1} \right)}
                \right] ^{ \left( -1 \right) ^{n}}.
                    \label{eq:I7a1b0c0d0case2int2}
                \end{align}
                \label{eq:I7a1b0c0d0case2int}
            \end{subequations}
        \end{description}
   \item[Case $a=-1,b=1, c=0$:]
       For this value of the parameters only degenerate or linear equations arise.
       In particular we have discrete wave-like equation.
    \item[Case $a=-1,b=0$:]
        For this value of the parameters we have several linear wave-like
        equations, but also one nontrivial equation:
        \begin{equation} 
            \left(u_{n+1,m+1}+\frac{d_3}{c_3}\right)\left(u_{n,m}-\frac{b_4}{c_3}\right)
            -\left(u_{n,m+1}+\frac{d_3}{c_3}\right)\left(u_{n+1,m}-\frac{b_4}{c_3}\right)=0,
            \label{(4)} 
        \end{equation}
        where we have the following relationship between the parameters:
        \begin{equation}
            d=-c\frac{b_4}{c_3},\quad c(b_4+d_3)=0.
            \label{eq4params}
        \end{equation}
        The case $b_4=-d_3$ \eqref{(4)} is equivalent to \eqref{ple} up to a linear transformation. 
        The case $b_4\neq -d_3$, $c=d=0,$ and \eqref{(4)} is equavalent to (4) of List 4 of 
        \cite{GarifullinYamilov2012} up to a linear transformation and it is a LT equation.
\end{description}

\paragraph{Equation \eqref{Bur}:}
To equation \eqref{Bur} correspond only degenerate quad-equations
in the sense of condition \eqref{eq:nondeg} and the trivial
linearizable equation \eqref{ple2}.

\subsection{List 3}

\paragraph{Equation \eqref{Vol1s}:}
To equation \eqref{Vol1s} correspond only degenerate quad-equations
in the sense of condition \eqref{eq:nondeg} and the trivial
linearizable equation \eqref{ple1}.

\paragraph{Equation \eqref{mVol2}:}
Equation \eqref{mVol2} depends on the parameters $a$ and $c$.
If $a\neq0$ then it can always be rescaled to 1, hence
we can consider two independent cases: $a=1$ and $a=0$.

However, to equation \eqref{mVol2} with $a=0$ correspond only degenerate 
quad-equations in the sense of condition \eqref{eq:nondeg} and the trivial
linearizable equation \eqref{ple1}.

On the other hand, the case when $a=1$ admits nontrivial solutions.
Indeed,  to equation \eqref{mVol2} with $a=1$ correspond eight
nontrivial equations. 
Four of them are valid for every value of $c$:
\begin{equation}
     \left( u_{{n+1,m}}+\sigma_1 \right)  \left( u_{{n,m}}-\sigma_1 \right)=
     \left( u_{{n+1,m+1}}-\sigma_2 \right)  \left( u_{{n,m+1}}+\sigma_2 \right),
     \quad \sigma_{i}=\pm1.
    \label{eq:II4a1cm4}
\end{equation}
Using the transformation $u_{n,m}\to -u_{n,m}$ we can reduce the
(\ref{eq:II4a1cm4},$-\sigma_{1}=\sigma_{2}=1$) case to the
(\ref{eq:II4a1cm4},$\sigma_{1}=-\sigma_{2}=1$) case and the
(\ref{eq:II4a1cm4},$\sigma_{1}=\sigma_{2}=-1$) case to the
(\ref{eq:II4a1cm4},$\sigma_{1}=\sigma_{2}=1$) case respectively.
This means that we have only the two following independent equations
\begin{subequations}
    \begin{gather}
        (u_{n+1,m}+1)(u_{n,m}-1)=(u_{n+1,m+1}-1)(u_{n,m+1}+1),
        \label{t1}
        \\
        (u_{n+1,m}+1)(u_{n,m}-1)=(u_{n+1,m+1}+1)(u_{n,m+1}-1).
        \label{(L4_3)}
    \end{gather}
    \label{eq:II4a1cm4expl}
\end{subequations}

Other four equations are found if $c=-4$:
\begin{equation}
     \left( u_{{n+1,m}}+\sigma_1 \right)  \left( u_{{n,m}}-\sigma_1 \right)=
     - \left( u_{{n+1,m+1}}-\sigma_2 \right)  \left( u_{{n,m+1}}+\sigma_2 \right),
     \quad \sigma_{i}=\pm1.
    \label{eq:II4a1cgen}
\end{equation}
Using the transformation $u_{n,m}\to -u_{n,m}$ we can reduce
the (\ref{eq:II4a1cgen},$-\sigma_{1}=\sigma_{2}=1$) case to the
(\ref{eq:II4a1cgen},$\sigma_{1}=-\sigma_{2}=1$) case and the
(\ref{eq:II4a1cgen},$\sigma_{1}=\sigma_{2}=-1$) case to the
(\ref{eq:II4a1cgen},$\sigma_{1}=\sigma_{2}=1$) case respectively.
This means that we have only the two following independent equations
\begin{subequations}
    \begin{gather}
        (u_{n+1,m}+1)(u_{n,m}-1)=-(u_{n+1,m+1}+1)(u_{n,m+1}-1),
        \label{(L4_3m)}
        \\
        (u_{n+1,m}+1)(u_{n,m}-1)=-(u_{n+1,m+1}-1)(u_{n,m+1}+1).
        \label{t1m}
    \end{gather}
    \label{eq:II4a1cgenexpl}
\end{subequations}

Equation \eqref{t1} is a well-known sGT equation.
It corresponds to equation $T1^*$ in \cite{LeviYamilov2011}.
See references therein for a discussion of the appearance of this 
equation in literature. 
On the other hand equation \eqref{(L4_3)} is transformed into equation
(3) of List 4 in \cite{GarifullinYamilov2012} with the linear 
transformation $u_{n,m}=2\hat{u}_{n,m}+1$.
This means that this equation is a known LT equation with the following
first integrals:
\begin{equation}
    W_{1} = \left( u_{n,m}-1 \right)u_{n+1,m}+u_{n,m},
    \quad
    W_{2} = {\frac { \left( u_{{n,m}}-u_{{n,m-1}} \right)  \left( u_{{n,m+2}}-u_{{n,m+1}} \right) }{%
    \left( u_{{n,m+1}}-u_{{n,m-1}} \right)  \left( u_{{n,m+2}}-u_{{n,m}} \right) }}.
    \label{eq:L4_3_int}
\end{equation}
For further details see \cite{Startsev2010}.

Equation \eqref{(L4_3m)}  is LT new example of Darboux integrable 
equation and has first integrals:
\begin{subequations}
    \begin{align}
        W_1 &=(-1)^m(u_{n+1,m}+1)(u_{n,m}-1),
        \\
        W_2 &=\frac{(u_{n,m+3}-u_{n,m-1})(u_{n,m+1}-u_{n,m-3})}{(u_{n,m+3}-u_{n,m+1})(u_{n,m-1}-u_{n,m-3})}.
    \end{align}
    \label{eq:L4_3m_int}
\end{subequations}
The integral $W_{2}$ is a sixth order (seven-point) first integral.
It can be proved that there is no first integral of lower order in this direction.
Examples of equations with first integrals of such a high order,
except for a special series of Darboux integrable equations presented 
in \cite{GarifullinYamilov2012Ufa}, are, at the best of our knowledge,
new in the literature.

Now we are going to use a transformation theory presented in \cite{Yamilov1991,Startsev2010}. 
We can rewrite \eqref{(L4_3m)} in a special form and introduce a new function 
$v_{n,m}$ as follows:
\begin{equation}
    \frac{u_{n+1,m}+1}{u_{n+1,m+1}+1}=-\frac{u_{n,m+1}-1}{u_{n,m}-1}=v_{n+1,m}.
\end{equation}
So we have
\begin{subequations}
    \begin{gather}
        v_{n,m}=\frac{u_{n,m}+1}{u_{n,m+1}+1},\quad v_{n+1,m}=-\frac{u_{n,m+1}-1}{u_{n,m}-1},
        \\
        u_{n,m}=1+2\frac{v_{n,m}-1}{v_{n+1,m}v_{n,m}+1},\quad u_{n,m+1}=-1+2\frac{v_{n+1,m}+1}{v_{n+1,m}v_{n,m}+1}.
    \end{gather}
    \label{eq:npinv}
\end{subequations}
This is a non-point transformation which is invertible on the solutions of \eqref{(L4_3m)}.
Formula \eqref{eq:npinv} allows to rewrite \eqref{(L4_3m)} as
\begin{equation}
    (v_{n+1,m+1}+1)(v_{n,m}-1)=(v_{n+1,m}^{-1}+1)(v_{n,m+1}^{-1}-1),
    \label{eq_v}
\end{equation}
and its first integals as:
\begin{subequations}
    \begin{align}
        &W_1=(-1)^m\frac{(v_{n+1,m}+1)v_{n,m}(v_{n-1,m}-1)}{(v_{n+1,m}v_{n,m}+1)(v_{n,m}v_{n-1,m}+1)},
        \\ 
        &W_2=\frac{(v_{n,m+3}v_{n,m+2}v_{n,m+1}v_{n,m}-1)(v_{n,m+1}v_{n,m}v_{n,m-1}v_{n,m-2}-1)}{(v_{n,m+3}v_{n,m+2}-1)v_{n,m+1}v_{n,m}(v_{n,m-1}v_{n,m-2}-1)}.
    \end{align}
    \label{W_v}
\end{subequations}
Both first integrals \eqref{W_v} are of lowest possible order in their directions.
Example \eqref{eq_v} is still exeptional, even though
the order of $W_2$ has been lowered to five. 

Equation \eqref{t1m} is a sGT equation, as we can prove that it
possesses the following second order autonomous generalized symmetries 
in both directions:
\begin{subequations}
    \begin{align}
        \dv{u_{n,m}}{t_{1} }&
        \begin{aligned}[t]
            &=(u_{n,m}^2-1)
            \big[(u_{n+1,m}^2-1)(u_{n+2,m}+u_{n,m})
            \\
            &\quad-(u_{n-1,m}^2-1)(u_{n,m}+u_{n-2,m})
            -4(u_{n+1,m}-u_{n-1,m})\big],
        \end{aligned}
        \label{sym_mN2}
        \\
        \dv{u_{n,m}}{t_{2}} &=(u_{n,m}^2-1)(T_m-1)
        \left(\frac{u_{n,m+1}+u_{n,m}}{U_{n,m}}+\frac{u_{n,m-1}+u_{n,m-2}}{U_{n,m-1}}\right),
        \label{sym_nN2}
    \end{align}
\end{subequations}
where
\begin{equation}
    U_{n,m}=(u_{n,m+1}+u_{n,m})(u_{n,m}+u_{n,m-1})-2(u_{n,m}^2-1).
    \label{eq:Ut1mdef}
\end{equation}
It can be proved that equation \eqref{t1m} does not admit
any autonomous generalized symmetries of lower order.
Equation \eqref{t1m} seems to be a new interesting example and deserves 
a separate study, see \cite{GarifullinYamilov2018}.

    Here we limit ourselves to observe that equation \eqref{t1m} 
    is an integrable equation according to the algebraic
    entropy test \cite{BellonViallet1999,Viallet2006,Tremblay2001}.
    Indeed, computing the degree of the iterates of equation
    \eqref{t1m} by using the library \texttt{ae2d.py} \cite{GubHay,GubbiottiPhD2017},
    we obtain the following sequence:
    \begin{equation}
        1, 2, 4, 7, 11, 16, 22, 29, 37, 46, 56, 67, 79\dots.
        \label{eq:t1mseq}
    \end{equation}
    The sequence \eqref{eq:t1mseq} is fitted by the following
    generating function:
    \begin{equation}
        g(z)= \frac{z^2 - z + 1}{(1-z)^3}.
        \label{eq:t1mgf}
    \end{equation}
    Since the generating function \eqref{eq:t1mgf} has only one
    pole in $z_{0}=1$ and it lies on the unit circle, equation 
    \eqref{t1m} is integrable according to the algebraic entropy criterion.
    Moreover, due to the presence of $\left( z-1 \right)^{3}$ in the denominator 
    of the generating function \eqref{eq:t1mgf} we have that growth
    of equation \eqref{t1m} is \emph{quadratic}.
    According to the classification of discrete equations using
    algebraic entropy \cite{HietarintaViallet2007} we have that equation
    \eqref{t1m} is supposed to be genuinely integrable and not
    linearizable.

\paragraph{Equation \eqref{Volz}:}
To equation \eqref{Volz} with $a\neq0$ corresponds one quad-equation:
\begin{equation}
    \begin{aligned}
        &\left(u_{n,m}-u_{n,m+1}+\frac{b_4}{c_3}\right)\left(u_{n+1,m}-u_{n+1,m+1}+\frac{b_4}{c_3}\right)
        \\ 
        +&\frac{b_4+d_3}{c_3}\left(u_{n,m+1}-u_{n+1,m}-\frac{b_4}{c_3}-a\right)=0, \quad b_4+d_3\neq 0.
    \end{aligned}
    \label{(t4)}
\end{equation}
When $a=0$, to equation \eqref{Volz} corresponds two quad-equations. 
One of them is just \eqref{(t4)} with $a=0$, and the second one is:
\begin{equation}
    \begin{aligned}
        &\left(u_{n,m}+u_{n,m+1}-\frac{b_4}{c_3}\right)\left(u_{n+1,m}+u_{n+1,m+1}-\frac{b_4}{c_3}\right)
        \\
        +&\frac{b_4-d_3}{c_3}\left(u_{n,m+1}+u_{n+1,m}-\frac{b_4}{c_3}\right)=0, \quad b_4-d_3\neq 0,
    \end{aligned}
    \label{(t4+)}
\end{equation}
but in this case we have $b=c=0$.

For any $a$ equation \eqref{(t4)} is transformed into (T4) of \cite{LeviYamilov2011} 
by the following non-autonomous linear point transformation:
\begin{equation}
    u_{n,m}=\hat u_{n,m}\frac{b_4+d_3}{c_3}-n\left(\frac{b_4+d_3}{2c_3}+a\right)+m\frac{b_4-d_3}{2c_3},
    \label{eq:tot4}
\end{equation}
where $\hat u_{n,m}$ satisfies (T4):
\begin{equation}
    \left( \hat{u}_{n,m}-\hat{u}_{n,m+1}+\frac{1}{2} \right)\left( \hat{u}_{n+1,m}-\hat{u}_{n+1,m+1}+\frac{1}{2} \right)
    +\hat{u}_{n,m+1}-\hat{u}_{n+1,m}=0.
    \label{eq:T42011}
    \tag{T4}
\end{equation}
We recall that equation \eqref{eq:T42011} is a sGT equation, with
two nontrivial three-point generalized symmetries \cite{LeviYamilov2011}.

In the same way equation \eqref{(t4+)} is another sGT equation which is transformed 
by linear transformation of $u_{n,m}$ into the special form: 
\begin{equation}
    \left(u_{n,m}+u_{n,m+1}+1\right)\left(u_{n+1,m}+u_{n+1,m+1}+1\right)
    =2\left(u_{n,m+1}+u_{n+1,m}+1\right). 
    \label{(t4+1)}
\end{equation}
This equation is a new sGT quad-equation which possesses five-point autonomous generalized symmetries:
\begin{subequations}
    \begin{align}
        \dv{u_{n,m}}{\theta_{1}}&=(u_{n+2,m}-u_{n-2,m})(u_{n+1,m}-u_{n,m})(u_{n,m}-u_{n-1,m}),
        \label{sym_n}
        \\
        \dv{u_{n,m}}{\theta_{2}} &=\frac{[(u_{n,m+1}+u_{n,m})^2-1][(u_{n,m}+u_{n,m-1})^2-1]}{u_{n,m+1}+2u_{n,m}+u_{n,m-1}}
        (T_m-\Id)\left(\frac{1}{U_{n,m}}\right),   
        \label{sym_m}
    \end{align}
    \label{eq:symmt4plus}
\end{subequations}
where
\begin{equation}
    U_{n,m}=(u_{n,m+1}+u_{n,m-1})(u_{n,m}+u_{n,m-2})+2(u_{n,m+1}u_{n,m-1}+u_{n,m}u_{n,m-2}+1).
    \label{eq:sym_mUnm}
\end{equation}

We note that equation \eqref{(t4+)} possesses also the following 
non-autonomous three-point generalized symmetry in the $n$-direction:
\begin{equation}
    \dv{u_{n,m}}{\theta_{3}} = \left( -1 \right)^{m}
    \left(u_{n+1,m}-u_{n,m}\right)\left(u_{n,m}-u_{n-1,m}\right),
    \label{eq:symm3_t4plus}
\end{equation}
and the non-autonomous point symmetry
\begin{equation}
    \dv{u_{n,m}}{\theta_{4}} = \left( -1 \right)^{m}.
    \label{eq:symm1_t4plus}
\end{equation}
On the other side it is possible to prove that no 
three-point generalized symmetry exists in the $m$-direction.
Hence the generalized symmetry \eqref{sym_m} is the lowest order generalized symmetry
of equation \eqref{(t4+)} in the $m$-direction.

%We did not checked that these symmetries were the simplest in their directions. 

We remark that equation \eqref{(t4+)} through the following nontrival
non-invertible transformation:
\begin{equation}
    \hat u_{n,m}=u_{n,m+1}+u_{n,m},
    \label{E5E4}
\end{equation}
is transformed into \eqref{t1m}, i.e. we have that $\hat u_{n,m}$ satisfies equation \eqref{t1m}.
The symmetries (\ref{sym_n},\ref{sym_m}) are transformed into symmetries 
(\ref{sym_mN2},\ref{sym_nN2}) too, but here $\theta_1=4t_1,\ \theta_2=2t_2$. 
However, in case of \eqref{sym_n}, we can check this fact only by using discrete equation \eqref{(t4+1)}. 
This means that only common solutions of \eqref{sym_n} and \eqref{(t4+1)} 
are transformed into solutions of \eqref{sym_mN2}.

Finally, we underline that equation \eqref{sym_m} is a particular case of 
known example of integrable five-point differential-difference equation 
presented in \cite{Adler2018}.
According to this remark we have that equation \eqref{sym_nN2} is a
nontrivial modification of a known example.

\paragraph{Equation \eqref{Vol_mod1s}:}
To equation \eqref{Vol_mod1} correspond only degenerate quad-equations
in the sense of condition \eqref{eq:nondeg} and the trivial
linearizable equation \eqref{ple1}.

\paragraph{Equation \eqref{mVol3}:}
Equation \eqref{mVol3} depends on the parameter $a$.
If $a\neq0$ then it can always be rescaled to 1, hence
we can consider two independent cases: $a=1$ and $a=0$.
However, in both cases to equation \eqref{mVol3} corresponds only 
degenerate quad-equations in the sense of condition \eqref{eq:nondeg}.

%\paragraph{Equation (\ref{mVol3}, $a=0$):}
%To equation \eqref{mVol3} with $a=0$ correspond only degenerate quad-equations
%in the sense of condition \eqref{eq:nondeg} and the trivial
%linearizable equation \eqref{ple1}.

\paragraph{Equation \eqref{Vol_mod2}:}
To equation \eqref{Vol_mod1} correspond only degenerate quad-equations
in the sense of condition \eqref{eq:nondeg} and the trivial
linearizable equation \eqref{ple1}.

\subsection{List 4}

\paragraph{Equation \eqref{our1}}
To equation \eqref{our1} correspond only degenerate quad-equations
in the sense of condition \eqref{eq:nondeg}.

\paragraph{Equation \eqref{our2}}
To equation \eqref{our2} correspond only degenerate quad-equations
in the sense of condition \eqref{eq:nondeg} and equations of the form
of the linear discrete wave equations \eqref{eq:dwave}.

\paragraph{Equation \eqref{ourii}:}
To equation \eqref{ourii} correspond two nontrivial
equations:
\begin{subequations}
    \begin{align}
        u_{n+1,m+1}(u_{n,m}-u_{n,m+1})-u_{n+1,m}(u_{n,m}+u_{n,m+1})+2=0,
        \label{(39)}
        \\
        u_{n+1,m}(u_{n,m+1}-u_{n,m})-u_{n+1,m+1}(u_{n,m+1}+u_{n,m})+2=0.
        \label{39bis}
    \end{align}
    \label{eq:II13eqs}
\end{subequations}
Equation \eqref{39bis} can be reduced to equation \eqref{(39)} using
the transformation \eqref{trmm}.
Using the scaling $u_{n,m}=\sqrt{2}\hat{u}_{n,m}$ equation \eqref{(39)}
is reduced to:
\begin{equation}
    \hat{u}_{n+1,m+1}(\hat{u}_{n,m}-\hat{u}_{n,m+1})-\hat{u}_{n+1,m}(\hat{u}_{n,m}+\hat{u}_{n,m+1})+1=0,
    \label{39orig}
\end{equation}
which is a known sGT equation.
Namely equation \eqref{39orig} is equation (39) of \cite{GarifullinYamilov2012}, see \cite{GarifullinMikhailovYamilov2014} for more details.
Indeed, from \cite{GarifullinYamilov2012} we know that actually
equation \eqref{(39)} has a more complicated symmetry structure
in the $n$ direction and its general five-point symmetry is the 
following \emph{non-autonomous} one:
\begin{equation}
    \dv{u_{n,m}}{\varepsilon_{1}}=
    \left( u_{n+1,m}u_{n,m}-1 \right)\left( u_{n,m}u_{n-1,m}-1 \right)
    \left( a_{n}u_{n+2,m}-a_{n-1}u_{n-2,m} \right),
    \label{eq:symm_n39gen}
\end{equation}
where $a_{n}$ is a two-periodic function, i.e. a solution of the
difference equation $a_{n+2}=a_{n}$.
We note that $a_{n}=1$ is a solution which corresponds to
equation \eqref{ourii}.
Moreover equation \eqref{(39)} has a non-autonomous three-point
symmetry in the $m$ direction:
\begin{equation}
    \dv{u_{n,m}}{\varepsilon_{2}} = \left( -1 \right)^{n}\frac{u_{n,m+1}u_{n,m-1}+u_{n,m}^{2}}{u_{n,m+1}+u_{n,m-1}}.
    \label{eq:symm_m39gen}
\end{equation}
No autonomous three-point symmetry exists for equation \eqref{(39)}.

\subsection{List 5}

\paragraph{Equation \eqref{INB}}
To equation \eqref{INB} correspond only degenerate quad-equations
in the sense of condition \eqref{eq:nondeg}.

\paragraph{Equation \eqref{mod_INB}}
Equation \eqref{mod_INB} depends on two parameters $a$ and $b$.
If $a\neq0$ it can always be scaled to $a=1$, so for this equations
we get two different cases: when $a=1$ and when $a=0$.

\begin{description}
    \item[Case $ a=1$:]
       For this value of the parameter only degenerate or linear equations arise.
       In particular we have a discrete wave equation \eqref{eq:dwave}.
    \item[Case $ a=0$:]
       For this value of the parameter only degenerate or linear equations arise.
       In particular we have a discrete wave equation \eqref{eq:dwave}.
\end{description}

\paragraph{Equation \eqref{mikh1}}
Equation \eqref{mikh1} depends only on the parameter $a$.
If $a\neq0$ it can be scaled to $a=1$ always.
Therefore we have to consider the cases $a=1$ and $a=0$.

\begin{description}
    \item[Case $a=1$:] 
        For this value of the parameter there are two nontrivial equations:
        \begin{subequations}
            \begin{align}
                u_{n+1,m+1}(u_{n,m+1}+u_{n,m}+1)+u_{n,m}(u_{n+1,m}+1)=0,
                \label{mx13}
                \\
                u_{n,m+1}(u_{n+1,m+1}+u_{n+1,m}+1)+u_{n+1,m}(u_{n,m}+1)=0.
                \label{mx13mod}
            \end{align}
        \end{subequations}
        Up to transformation \eqref{trmm} we have that equation
        \eqref{mx13mod} reduces to equation \eqref{mx13}.
        Therefore we have only one independent equation.
        Equation \eqref{mx13} is a sGT equation and appeared as formula (58) in \cite{MikhailovXenitidis2013}.
    \item[Case $a=0$:] 
        For this value of the parameter there are six nontrivial equations:
        \begin{subequations}
            \begin{gather}
               \left(u_{{n,m}}+u_{{n,m+1}}\right)u_{{n+1,m+1}}+u_{{n,m}}u_{{n+1,m}}=0,
               \label{(6)_4}
                \\
                \left(u_{{n,m}}+u_{{n,m+1}}\right)u_{{n+1,m}}+u_{{n,m+1}}u_{{n+1,m+1}}=0,
               \label{(6)_4b}
                \\
                \left(\frac{1+\imath\sqrt {3}}{2}u_{{n,m}}+\frac{1-\imath\sqrt{3}}{2}u_{{n,m+1}}\right)u_{{n+1,m+1}}
            - u_{{n,m}}u_{{n+1,m}}=0,
               \label{(6)_4c}
                \\
                \left(\frac{1-\imath\sqrt {3}}{2}u_{{n+1,m}}+\frac{1+\imath\sqrt{3}}{2}u_{{n+1,m+1}}\right)u_{{n,m+1}}
                -u_{{n,m}}u_{{n+1,m}}=0,
               \label{(6)_4d}
                \\
                \left(\frac{1+\imath\sqrt {3}}{2}u_{{n+1,m}}+\frac{1-\imath\sqrt {3}}{2}u_{{n+1,m+1}}\right)u_{{n,m+1}}
                -u_{{n,m}}u_{{n+1,m}}=0,
               \label{(6)_4e}
                \\
                \left(\frac{1-\imath\sqrt {3}}{2}u_{{n,m}}+\frac{1+\imath\sqrt {3}}{2}u_{{n,m+1}}\right)u_{{n+1,m+1}}
                -u_{{n,m}}u_{{n+1,m}}=0.
               \label{(6)_4f}
            \end{gather}
            \label{eq:I13a0}
        \end{subequations}
        Up to transformations \eqref{trmsm}, \eqref{trmsm3} and\eqref{trmm}
        we have that all the equations in \eqref{eq:I13a0} can be reduced to \eqref{(6)_4}.
        The precise relationship between these equations is illustrated in
        Figure \ref{fig:I13transf}.
        Equation \eqref{(6)_4} is a known LT equation and it can be identified
        with (6) with $a_2=1$ of List 4 of \cite{GarifullinYamilov2012}.
        Its first integrals are:
        \begin{equation}
            W_{1} = u_{n+1,m}u_{n,m}u_{n-1,m},
            \quad
            W_{2} = d^{-n}\frac{u_{n,m}-d u_{n,m+1}}{d u_{n,m}-u_{n,m+1}},
            \quad
            d \equiv -\frac{1+\imath \sqrt{3}}{2}.
            \label{eq:intI13a0}
        \end{equation}

        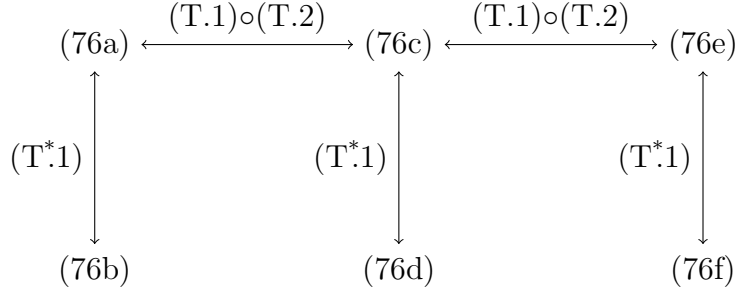
\begin{figure}[hbtp]
            \centering
            \begin{tikzpicture}
                \path (0,3) node(e1) {\eqref{(6)_4}}
                    (0,0) node(e2) {\eqref{(6)_4b}}
                    (4,3) node(e3) {\eqref{(6)_4c}}
                    (4,0) node(e4) {\eqref{(6)_4d}}
                    (8,3) node(e5) {\eqref{(6)_4e}}
                    (8,0) node(e6) {\eqref{(6)_4f}};
                    \draw[<->] (e1) -- node[left] {\eqref{trmm}} (e2);
                    \draw[<->] (e3) -- node[left] {\eqref{trmm}} (e4);
                    \draw[<->] (e5) -- node[left] {\eqref{trmm}} (e6);
                    \draw[<->] (e1) -- node[above] {\eqref{trmsm}$\circ$\eqref{trmsm3}} (e3);
                    \draw[<->] (e3) -- node[above] {\eqref{trmsm}$\circ$\eqref{trmsm3}} (e5);
                %\draw (0,6) node[circle,draw] {\eqref{(6)_4}}
                %--(0,4) node[circle,draw] {\eqref{(6)_4c}};
            \end{tikzpicture}
            \caption{The relationship between the equations \eqref{eq:I13a0}.
            By \eqref{trmsm}$\circ$\eqref{trmsm3} we mean the composition of the
            two transformations.}
            \label{fig:I13transf} 
        \end{figure}

\end{description}

\paragraph{Equation \eqref{eq1}:}
To equation \eqref{eq1} correspond only degenerate quad-equations
in the sense of condition \eqref{eq:nondeg}.

\paragraph{Equation \eqref{eq311}:}
To equation \eqref{eq311} correspond three nontrivial 
quad-equations\footnote{By direct computations several sub-cases of
    equations \eqref{eq:eq311tot} are found. We omit them since they
are not independent cases.}:
\begin{subequations}
    \begin{gather}
        \left(u_{n+1,m+1}+\frac{c_4}{a_4}\right)\left(u_{n,m+1}+\frac{c_4}{a_4}\right)-\left(u_{n+1,m}-\frac{d_3}{a_4}\right)\left(u_{n,m}-\frac{d_3}{a_4}\right)=0,\quad a_4\neq0,
        \label{l4_1}
        \\
        \left(u_{n+1,m+1}+u_{n+1,m}+\frac{c_4}{c_3}\right)\left(u_{n,m+1}+u_{n,m}+\frac{c_4}{c_3}\right)=\frac{c_4^2-c_3d_4}{c_3^2},
        \label{l4_5_1} \quad c_3\neq0,\, c_4^2-c_3d_4\neq0,
        \\
        \left(u_{n+1,m+1}-u_{n+1,m}-\frac{c_4}{c_3}\right)
        \left(u_{n,m+1}-u_{n,m}-\frac{c_4}{c_3}\right)=
        \frac{c_4^2+c_3d_4}{c_3^2},\quad \label{l4_5} c_3\neq0,\, c_4^2+c_3d_4\neq 0.
    \end{gather}
    \label{eq:eq311tot}
\end{subequations}

Equation \eqref{l4_1} is a LT equation. 
Using the transformation 
\begin{equation}
    \hat u_{n,m}=u_{n,m}+\frac{c_4}{a_4}
    \label{eq:T311_1}
\end{equation}
we can bring it into equation (1) of List 4 of \cite{GarifullinYamilov2012}.
Its first integrals are given by:
\begin{equation}
    W_{1} = (u_{n+2, m}-u_{n, m})(u_{n+1, m}-u_{n-1, m}),
    \quad
    W_{2}=
    \left( \frac {a_4u_{{n,m+1}}+c_4}{a_4u_{{n,m}}- d_3} \right) ^{ \left( -1 \right) ^{n}}.    
    \label{eq:int311a}
\end{equation}

Equation \eqref{l4_5_1} under the non-autonomous transformation
\begin{equation}
    \hat u_{n,m}=(-1)^m\left(u_{n,m}+\frac{c_4}{2c_3}\right),
    \label{eq:T311_2}
\end{equation}
is mapped into a particular case of equation \eqref{l4_5}. 
Equation \eqref{l4_5} itself is of LT. It can be identified
with equation (5) of List 4 of \cite{GarifullinYamilov2012}.
This implies that its first integrals are given by:
\begin{equation}
    W_{1} = u_{n+1,m}-u_{n-1,m},
    \quad
    W_{2} = \left( -1 \right)^{n}
    \frac{u_{n,m+1}-u_{n,m}+\kappa-\delta}{u_{n,m+1}-u_{n,m}+\kappa+\delta},
    \,
    \kappa=\frac{c_{4}}{c_{3}},
    \,
    \delta=\frac{c_4^2+c_3d_4}{c_3^2}.
    \label{eq:int311}
\end{equation}

%\eq{u_{n+1,m+1}(u_{n,m}+u_{n,m+1})+u_{n+1,m}u_{n,m}=0.}

%\\
\paragraph{Equation \eqref{INB1}:}
To equation \eqref{INB1} correspond only degenerate quad-equations
in the sense of condition \eqref{eq:nondeg} and the 
trivial linearizable equation \eqref{ple1}.

\paragraph{Equation \eqref{INB3}:}
Equation \eqref{INB3} depends on the parameter $a$.
If $a\neq0$ then it can always be rescaled to 1, hence
we can consider two independent cases: $a=1$ and $a=0$.
However in both cases to equation \eqref{INB1} correspond only 
degenerate quad-equations in the sense of condition \eqref{eq:nondeg} or
linear equations of the form of the discrete wave equation
\eqref{eq:dwave}.

\paragraph{Equation  \eqref{MX2}:}
Equation \eqref{MX2} depends on the parameter $a$.
If $a\neq0$ then it can always be rescaled to 1, hence
we can consider two independent cases: $a=1$ and $a=0$.

To equation \eqref{MX2} when $a=1$ correspond four nontrivial
equations:
\begin{subequations}
    \begin{align}
        u_{{n,m}}u_{{n+1,m}}+u_{{n,m+1}}u_{{n+1,m}}+u_{{n,m+1}}u_{{n+1,m+1}} &=1,
        \label{shl}
        \\
        u_{{n,m}}u_{{n+1,m}}-u_{{n,m+1}}u_{{n+1,m}}+u_{{n,m+1}}u_{{n+1,m+1}} &=1,
        \label{eq:II11a1b}
        \\
        u_{{n,m}}u_{{n+1,m}}+u_{{n,m}}u_{{n+1,m+1}}+u_{{n,m+1}}u_{{n+1,m+1}} &=1,
        \label{eq:II11a1c}
        \\
        u_{{n,m}}u_{{n+1,m}}-u_{{n,m}}u_{{n+1,m+1}}+u_{{n,m+1}}u_{{n+1,m+1}} &=1.
        \label{eq:II11a1d}
    \end{align}
    \label{eq:II11a1}
\end{subequations}
Up to transformations (\ref{trmm},\ref{trmsm}) we have that all the
equations in \eqref{eq:II11a1} can be reduced to \eqref{shl},
see Figure \ref{fig:II11a1transf}.
Equation \eqref{shl} is a known equation.
It was introduced in \cite{HernandezLeviScimiterna2013},
while in \cite{ScimiternaHayLevi2014} it was proved that it
is a sGT equation\footnote{Therein it is given by equation (1.9).}.

\begin{figure}[hbtp]
    \centering
    \begin{tikzpicture}
        \path (0,3) node(e1) {\eqref{shl}}
            (4,3) node(e2) {\eqref{eq:II11a1b}}
            (0,0) node(e3) {\eqref{eq:II11a1c}}
            (4,0) node(e4) {\eqref{eq:II11a1d}};
            \draw[<->] (e1) -- node[above] {\eqref{trmsm}} (e2);
            \draw[<->] (e1) -- node[left] {\eqref{trmm}} (e3);
            \draw[<->] (e1) -- node[right] {\eqref{trmm}$\circ$\eqref{trmsm}} (e4);
    \end{tikzpicture}
    \caption{The relationship between the equations \eqref{eq:II11a1}.
    By \eqref{trmm}$\circ$\eqref{trmsm} we mean the composition of the
    two transformations.}
    \label{fig:II11a1transf} 
\end{figure}
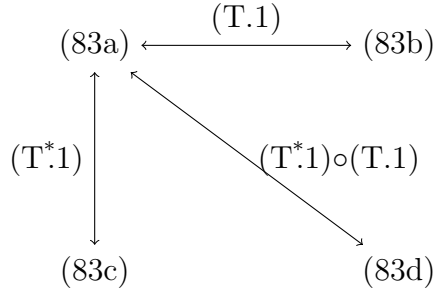

%\paragraph{Equation  (\ref{MX2}, $a=0$):}

To equation \eqref{MX2} when $a=0$ correspond twelve nontrivial 
equations.
Six of them are those of presented in formula \eqref{eq:I13a0}, but we
also have the following six other equations:
\begin{subequations}
    \begin{gather}
        u_{{n,m+1}}u_{{n+1,m+1}}+u_{{n+1,m}} \left( u_{{n,m}}-u_{{n,m+1}} \right)=0,        
        \label{eq:II11a0_1}
        \\
        \left( u_{{n,m}}-u_{{n,m+1}} \right) u_{{n+1,m+1}}-u_{{n,m}}u_{{n+1,m}}=0,
        \label{eq:II11a0_2}
        \\
        \left(  \frac{1-\imath \sqrt {3}}{2} u_{{n,m+1}}+u_{{n,m}} \right) u_{{n+1,m}}
        -\frac{1+\imath \sqrt {3}}{2} u_{{n,m+1}}u_{{n+1,m+1}}=0,
        \label{eq:II11a0_3}
        \\
        \left( u_{{n,m+1}}+ \frac{1-\imath \sqrt {3}}{2} u_{{n,m}} \right) u_{{n+1,m+1}}
        -\frac{1+\imath \sqrt {3}}{2} u_{{n+1,m}}u_{{n,m}}=0,
        \label{eq:II11a0_6}
        \\
        \left(  \frac{1+\imath \sqrt {3}}{2} u_{{n,m+1}}+u_{{n,m}} \right) u_{{n+1,m}}        
        -\frac{1-\imath \sqrt {3}}{2} u_{{n,m+1}}u_{{n+1,m+1}}=0,
        \label{eq:II11a0_4}
        \\
        \left(  \frac{1-\imath \sqrt {3}}{2} u_{{n,m+1}}+u_{{n,m}} \right) u_{{n+1,m+1}}
        + \frac{1+\imath \sqrt {3}}{2} u_{{n+1,m}}u_{{n,m}}=0.
        \label{eq:II11a0_5}
    \end{gather}
    \label{eq:II11a0}
\end{subequations}
Using the transformations (\ref{trmm},\ref{trmsm},\ref{trmsm3}) we have 
that the only independent equation is \eqref{(6)_4}.
See Figure \ref{fig:II11a0transf} for the details about the needed transformations.
As discussed above equation \eqref{(6)_4} is a LT equation with 
known first integrals \eqref{eq:intI13a0}.

\begin{figure}[hbtp]
    \centering
    \begin{tikzpicture}
        \path (0,3) node(e1) {\eqref{eq:II11a0_1}}
            (0,0) node(e2) {\eqref{eq:II11a0_2}}
            (4,3) node(e3) {\eqref{eq:II11a0_3}}
            (8,3) node(e4) {\eqref{eq:II11a0_4}}
            (8,0) node(e5) {\eqref{eq:II11a0_5}}
            (4,0) node(e6) {\eqref{eq:II11a0_6}}
            (4,-2) node (e0) {\eqref{(6)_4}};
            \draw[<->] (e1) -- node[left] {\eqref{trmm}} (e2);
            \draw[<->] (e3) -- node[left] {\eqref{trmm}} (e6);
            \draw[<->] (e4) -- node[left] {\eqref{trmm}} (e5);
            \draw[<->] (e1) -- node[above] {\eqref{trmsm}$\circ$\eqref{trmsm3}} (e3);
            \draw[<->] (e3) -- node[above] {\eqref{trmsm}$\circ$\eqref{trmsm3}} (e4);
        %\draw (0,6) node[circle,draw] {\eqref{(6)_4}}
        %--(0,4) node[circle,draw] {\eqref{(6)_4c}};
            \draw[<->,dotted] (e2) ..  controls (2,-1.5) .. node[dotted,above] {\eqref{trmsm}} (e0);
    \end{tikzpicture}
    \caption{The relationship between the equations \eqref{eq:II11a0} and \eqref{(6)_4}.
    By \eqref{trmsm}$\circ$\eqref{trmsm3} we mean the composition of the
    two transformations.}
    \label{fig:II11a0transf} 
\end{figure}
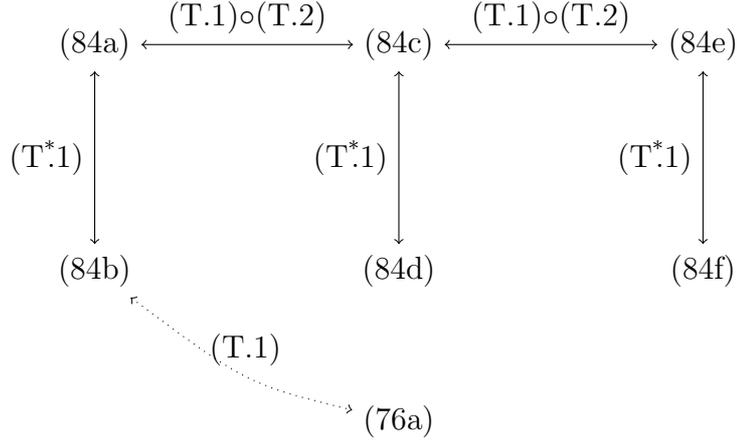

\paragraph{Equation \eqref{INB2}:}
To equation \eqref{INB2} correspond only degenerate quad-equations
in the sense of condition \eqref{eq:nondeg} or the trivial linearizable
equation \eqref{ple1}.

\subsection{List 6}

\paragraph{Equation \eqref{seva}:}
To equation \eqref{seva} correspond only degenerate quad-equations
in the sense of condition \eqref{eq:nondeg}.

\paragraph{Equation \eqref{rat}:}
To equation \eqref{rat} correspond only degenerate quad-equations
in the sense of condition \eqref{eq:nondeg}.

\paragraph{Equation \eqref{sroot}:}
Equation \eqref{sroot} is not rational, so we are not going to
consider it as we are dealing with multi-affine quad-equations
\eqref{Fh}.
It can be remarked that equation \eqref{sroot} has as a rational form, 
but it is quadratic in $u_{k\pm2}$, so it is outside the
main class of differential-difference equations.

\paragraph{Equation \eqref{SK2}:}
To equation \eqref{SK2} correspond only degenerate quad-equations
in the sense of condition \eqref{eq:nondeg} or equations of the form
of the linear discrete wave equation \eqref{eq:dwave}.

\section{Summary and outlook}
\label{sec:summary}

In this paper we have constructed all the possible autonomous
quad-equations \eqref{eq:quadgen} admitting, as generalized
symmetries, some known five-point differential-difference equations
belonging to a class recently classified in
\cite{GarifullinYamilovLevi2016,GarifullinYamilovLevi2018}.

In section \ref{sec:method} we gave a general result on all the
quad-equations admitting a generalized symmetry of the form \eqref{e1}, resulting
in the simplified form \eqref{Fh}.
Then in section \ref{sec:results} we considered all the explicit examples
of five-point differential-difference equations and found all the corresponding 
quad-equations.
We discussed the properties of these equations, highlighting which are
 LT and which are sGT equations.
We present a summary of these properties in Table \ref{tab:summary}.
Referring to the table, we have that the majority of the 
differential-difference equations of the form \eqref{eq:4thgen} as classified
in \cite{GarifullinYamilovLevi2016,GarifullinYamilovLevi2018} give raise to
degenerate or trivial quad-equations.
However eight of them give also rise to LT equations.
Most of these LT equations were known, but we also obtained two new
LT equations, namely equations \eqref{nDar} and \eqref{(L4_3m)}. 
Moreover, as expected, sGT equations are even rarer, appearing only in
in four examples.
Among these sGT equations there are two new examples, namely, equations \eqref{t1m} and \eqref{(t4+1)}.
In this paper we limited ourselves to present the result of the
algebraic entropy test for  equation \eqref{t1m}, which suggests integrability.
More details and generalizations, including $L-A$ pairs, 
on equation \eqref{t1m} can be found in \cite{GarifullinYamilov2018}.

As we remarked at the beginning of this section, in this paper we only
dealt with autonomous differential-difference equations and autonomous
quad-equations.
We are now working on weakening this condition for the fully discrete
equations to present a classification of non-autonomous quad-equations
admitting the differential-difference equations found in 
\cite{GarifullinYamilovLevi2016,GarifullinYamilovLevi2018} as five-point
generalized symmetries.
This classification can be performed with the method presented in
\cite{GarifullinYamilov2015} and applied to some known autonomous and
non-autonomous three-point differential-difference equations of
Volterra and Toda type.

\begin{table}
    \centering
    \begin{tabular}{cccc}
        \toprule
        \multicolumn{2}{c}{\textbf{Class I}} & \multicolumn{2}{c}{\textbf{Class 2}}
        \\
        \midrule
        Equations & Properties & Equations & Properties
        \\
        \midrule
        \eqref{Vol} & Trivial & \eqref{eq1ii} & Trivial
        \\
        \eqref{Vol0} & LT & \eqref{eq2ii} & LT\textsuperscript{*}
        \\
        \eqref{Vol1} & Trivial & \eqref{Vol1s} & Trivial
        \\
        \eqref{Vol_mod} & Trivial & (\ref{mVol2},$a=1$) & sGT\textsuperscript{*}, LT\textsuperscript{*}
        \\
        (\ref{Vol_mod1}, $a=1$) & Trivial & (\ref{mVol2},$a=0$) & Trivial
        \\
        (\ref{Vol_mod1}, $a=0$) & Trivial & \eqref{Volz} & sGT\textsuperscript{*}
        \\
        \eqref{Vol2} & LT & \eqref{Vol_mod1s} & Trivial
        \\
        (\ref{Bur2}, $a\neq0,\pm1,b=d=0,c\neq0$) & Trivial & (\ref{mVol3}, $a=1$) & Trivial
        \\
        (\ref{Bur2}, $a=1,b=c=d=0$) & LT  & (\ref{mVol3}, $a=0$) & Trivial
        \\
        (\ref{Bur2}, $a=-1,b=1,c=0$) & Trivial & \eqref{Vol_mod2} & Trivial
        \\
        (\ref{Bur2}, $a=-1,b=0$) & LT & \eqref{INB1} & Trivial
        \\
        \eqref{Bur} & Trivial & (\ref{INB3}, $a=1$) & Trivial
        \\
        \eqref{our1} & Trival & (\ref{INB3}, $a=0$) & Trivial
        \\
        \eqref{our2} & Trivial & (\ref{MX2}, $a=1$) & sGT
        \\
        \eqref{INB} & Trivial & (\ref{MX2}, $a=0$) & LT
        \\
        (\ref{mod_INB}, $a=1$) & Trivial & \eqref{INB2} & Trivial
        \\
        (\ref{mod_INB}, $a=0$) & Trivial & \eqref{ourii} & sGT
        \\
        (\ref{mikh1}, $a=1$) & sGT & \eqref{SK2} & Trivial
        \\
        (\ref{mikh1}, $a=0$) & LT & 
        \\
        \eqref{eq1} & Trivial & 
        \\
        \eqref{seva} & Trivial
        \\
        \eqref{rat} & Trivial
        \\
        \eqref{sroot} & Not rational
        \\
        \eqref{eq311} & LT
        \\
        \bottomrule
    \end{tabular}
    \caption{Summary of the properties of the corresponding discrete quad-equations.
    With \textsuperscript{*} we underline the presence of new quad-equations.}
    \label{tab:summary}
\end{table}

\section*{Acknowledgment}

GG is supported by the Australian Research Council through Nalini Joshi's
Australian Laureate Fellowship grant FL120100094.

\clearpage

\bibliographystyle{plain}
\bibliography{bibliography}

\end{document}